\documentclass[11pt,letterpaper]{article}
\usepackage[margin=1in]{geometry}
\usepackage[colorlinks]{hyperref}
\usepackage{tikz}
\usetikzlibrary{patterns}

\usepackage{amsfonts,amssymb,amsmath,amsthm}

\usepackage[utf8]{inputenc}
\usepackage{bibentry}
\usepackage{enumerate}

\newcommand{\cone}{\textup{cone}}
\newcommand{\aff}{\textup{aff}}
\newcommand{\lin}{\textup{lin}}
\newcommand{\rank}{\textup{rank}}
\newcommand{\convexhull}{\textup{conv}}
\newcommand{\dom}{\textup{\textup{dom}}}
\newcommand{\lattice}{\textup{\textup{lattice}}}
\newcommand{\zpmset}{\ensuremath{\{0,\pm 1\}}}
\newtheorem{theorem}{Theorem}[section]
\newtheorem{csq}[theorem]{Consequence}
\newtheorem{proposition}[theorem]{Proposition}
\newtheorem{obs}[theorem]{Observation}
\newtheorem{remark}[theorem]{Remark}

\newtheorem{conjecture}[theorem]{Conjecture}
\newtheorem{lemma}[theorem]{Lemma}
\newtheorem{definition}[theorem]{Definition}

\newtheorem{corollary}[theorem]{Corollary}
\newtheorem{open}[theorem]{Open Problem}

\date{\today}

\title{Principally Box-integer Polyhedra and Equimodular Matrices}%\thanks{This work was partially supported by the {\em } of the {\em }}}

\author{Patrick Chervet 
\thanks{Lycée Olympe de Gouges, rue de Montreuil à Claye, 93130, Noisy le Sec, France. 
E-mail: {\tt patrick.chervet@ac-creteil.com}}
\and Roland Grappe 
\thanks{Universit\'e Paris 13, Sorbonne Paris Cit\'e, LIPN, CNRS UMR 7030, F-93430, Villetaneuse, France. 
E-mail: {\tt roland.grappe@lipn.univ-paris13.fr}}
\and Louis-Hadrien Robert
\thanks{Universit\'e de Gen\`eve, rue du Li\`evre 2--4, 1227, Gen\`eve, Switzerland.
E-mail: {\tt louis-hadrien.robert@unige.ch}}
}

\begin{document}
\maketitle

\begin{abstract}
A polyhedron is {\em box-integer} if its intersection with any integer box $\{\ell\leq x \leq u\}$ is integer.
We define {\em principally box-integer} polyhedra to be the polyhedra $P$ such that $kP$ is box-integer whenever $kP$ is integer. We characterize them in several ways, involving equimodular matrices and box-total dual integral (box-TDI) systems. 
A rational $r\times n$ matrix is {\em equimodular} if it has full row rank and its nonzero $r\times r$ determinants all have the same absolute value. A {\em face-defining} matrix is a full row rank matrix describing the affine hull of a face of the polyhedron. Box-TDI systems are systems which yield strong min-max relations, and the underlying polyhedron is called a {\em box-TDI polyhedron}.
Our main result is that the following statements are equivalent.
\begin{itemize}
	\item[\textbullet] The polyhedron $P$ is principally box-integer.
	\item[\textbullet] The polyhedron $P$ is box-TDI.
	\item[\textbullet] Every face-defining matrix of $P$ is equimodular.
	\item[\textbullet] Every face of $P$ has an equimodular face-defining matrix.
	\item[\textbullet] Every face of $P$ has a totally unimodular face-defining matrix.
	\item[\textbullet] For every face $F$ of $P$, $\lin(F)$ has a totally unimodular basis.
\end{itemize}
Along our proof, we show that a cone $\{x:Ax\leq \mathbf{0}\}$ is box-TDI if and only if it is box-integer, and that these properties are passed on to its polar.

We illustrate their use by reviewing well known results about box-TDI polyhedra. We also provide several applications. The first one is a new perspective on the equivalence between two results about binary clutters. Secondly, we refute a conjecture of Ding, Zang, and Zhao about box-perfect graphs. Thirdly, we discuss connections with an abstract class of polyhedra having the Integer Carath\'eodory Property. Finally, we characterize the box-TDIness of the cone of conservative functions of a graph and provide a corresponding box-TDI system.
\end{abstract}

\tableofcontents

%%%%%%%%%%%%%%%%%%%%%%%%%%%%%%%%%%%%%%%%%%%%%%%%%%%%%%%%%%%%%%%%%%%%%%%%%%%%%%%%%%%%%%%%%%%%%%%%%%%%%%%%%%%%%%%%%%%%%%%%%%%%%%%%%%%%%%%%
%%%%%%%%%%%%%%%%%%%%%%%%%%%%%%%%%%%%%%%%%%%%%%%%%%%%%%%%%%%%%%%%%%%%%%%%%%%%%%%%%%%%%%%%%%%%%%%%%%%%%%%%%%%%%%%%%%%%%%%%%%%%%%%%%%%%%%%%

%%%%%%%%%%%%%%%%%%%%%%%%%%%%%%%%%%%%%%%%%%%%%%%%%%%%%%%%%%%%%%%%%%%%%%%%%%%%%%%%%%%%%%%%%%%%%%%%%%%%%%%%%%%%%%%%%%%%%%%%%%%%%%%%%%%%%%%%%%%%%%%%%%%%
%%%%%%%%%%%%%%%%%%%%%%%%%%%%%%%%%%%%%%%%%%%%%%%%%%%%%%%%%%%%%%%%%%%%%%%%%%%%%%%%%%%%%%%%%%%%%%%%%%%%%%%%%%%%%%%%%%%%%%%%%%%%%%%%%%%%%%%%%%%%%%%%%%%%
%%%%%%%%%%%%%%%%%%%%%%%%%%%%%%%%%%%%%%%%%%%%%%%%%%%%%%%%%%%%%%%%%%%%%%%%%%%%%%%%%%%%%%%%%%%%%%%%%%%%%%%%%%%%%%%%%%%%%%%%%%%%%%%%%%%%%%%%%%%%%%%%%%%%
\section{Introduction}\label{s:intro}

In this paper, we introduce an abstract class of polyhedra which have strong integrality properties, and we call them principally box-integer---see Definition~\ref{def:pbi} below.
As we shall see, this class contains important and well studied polyhedra in combinatorial optimization and integer programming, such as those described by a totally unimodular matrix~\cite{hokr}, polymatroids~\cite{ed}, and box-totally dual integral polyhedra~\cite{co}.

We provide several characterizations of principally box-integer polyhedra. In this regard, some matrices play an important role.  They generalize unimodular matrices and we call them equimodular matrices---see Definition~\ref{def:em} below. 
These matrices are studied under the name of matrices with the Dantzig property in~\cite{heto} or as unimodular sets of vectors in~\cite{he}. We show that the notion of principal box-integrality is strongly intertwined with that of equimodularity: equimodular matrices are characterized using principal box-integrality and, in turn, principally box-integer polyhedra are characterized by the equimodularity of a family of matrices.

These notions shed new lights on fundamental results in combinatorial optimization and integer programming. For instance, the classical characterization of unimodular matrices by Veinott and Dantzig~\cite{veda} and that of totally unimodular matrices due to Hoffman and Kruskal~\cite{hokr} can be reformulated and extended using principally box-integer polyhedra. More importantly, these notions bring a geometric and matricial perspective about the so-called box-totally dual integral systems. These systems are useful to prove strong min-max combinatorial theorems and are known to be difficult to handle. We prove that a polyhedron is principally box-integer if and only if it can be described by a box-totally dual integral system. This provides several new characterizations of the latter. We believe that these characterizations fill ``the lack of a proper tool for establishing box-total dual integrality''---to quote Ding, Tan, and Zang~\cite{ditaza}---and we illustrate their use.

\medskip

Before going deeper into the details of our contributions, let us give the main definitions relevant to this paper and review related results from the literature.

A polyhedron $P=\{x:Ax\leq b\}$ of $\mathbb{R}^{n}$ is {\em integer} if each of its faces contains an integer point and {\em box-integer} if $P\cap\{\ell\leq x \leq u\}$ is integer for all $\ell,u\in\mathbb{Z}^n$.
For $k\in\mathbb{Z}_{\scriptscriptstyle>0}$, the $k$th dilation of $P$ is $kP=\{kx: x\in P\}=\{x:Ax\leq kb\}$.

\begin{definition}\label{def:pbi}
A polyhedron $P$ is {\em principally box-integer} if $kP$ is box-integer for all $k\in\mathbb{Z}_{\scriptscriptstyle>0}$ such that $kP$ is integer.
\end{definition}

A full row rank $r\times n$ matrix is {\em unimodular} if it is integer and its nonzero $r\times r$ determinants have value $1$ or $-1$~\cite[Page~267]{sc}. There is a strong connection between principally box-integer polyhedra and the following generalization of unimodular matrices. 

\begin{definition}\label{def:em}
A rational $r\times n$ matrix is {\em equimodular} if it has full row rank and its nonzero $r\times r$ determinants all have the same absolute value.
\end{definition}

\paragraph{Unimodular matrices.}
The notion of unimodularity dates back to Smith \cite{sm} and ensures that a linear system has an integral solution for each integer right-hand side.
%Unimodular matrices are characterized by Veinott and Dantzig in~\cite{veda} by means of the integrality of a family of polyhedra.
Hoffman and Kruskal~\cite{hokr} proved that integral solutions still exist under the weaker condition that (*) the gcd of the $r\times r$ determinants equals $1$. 
Condition (*) and equimodularity are complementary generalizations of unimodularity, in the sense that if an integer matrix is equimodular and satisfies (*), then it is unimodular.
Hoffman and Oppenheim~\cite{hoop} introduced variants of unimodularity, which were afterward studied by Truemper~\cite{tr}.
In~\cite{beha,he}, it is proved that equimodular matrices ensure that all basic solutions are integer, as soon as one of them is---see also Barnett~\cite[Chap.~7]{ba}.

The stronger notion of total unimodularity plays a central role in combinatorial optimization.
A matrix is {\em totally unimodular} when all its subdeterminants have value in $\{0,\pm1\}$. 
Examples of such matrices are network matrices and incidence matrices of bipartite graphs.
Hoffman and Kruskal~\cite{hokr} characterized totally unimodular matrices to be the matrices for which the associated polyhedra are all box-integer. 
Several other characterizations were obtained since then---see {\em e.g.}~\cite{ca} and~\cite{ghho}. 
Totally unimodular matrices are now well understood thanks to the decomposition theorem of Seymour~\cite{se}. 
For a survey of related results, we refer to~\cite[Chap.~4 and~19]{sc}.
More recently, Appa~\cite{ap} and Appa and Kotnyek~\cite{apko} generalized total unimodularity to rational matrices, their goal being to ensure the integrality of the associated family of polyhedra for a specified set of right-hand sides, such as those with only even coordinates.
In another direction, Lee~\cite{le} generalized totally unimodular matrices by considering the associated linear spaces. The connections between his results and the previous ones are discussed in Kotnyek's thesis~\cite[Chap.~11]{ko}.

We will see how principal box-integrality fits within the characterization of unimodular matrices by Veinott and Dantzig~\cite{veda} and that of totally unimodular matrices due to Hoffman and Kruskal~\cite{hokr}. Then, these results are naturally extended to characterize equimodular matrices. 
Also, a new generalization of totally unimodular matrices appears in Section~\ref{ss:charactEM}, the notion of totally equimodular matrices, which still have nice polyhedral properties.

\paragraph{Box-integrality.}
In combinatorial optimization and integer programming, a desirable property for polyhedra is to be integer, as then the vertices can be seen as combinatorial objects. Henceforth, many results in those fields are devoted to the study of properties and descriptions of integer polyhedra. The stronger property of being box-integer is far less studied. Nevertheless, some important classes of polyhedra are known to be box-integer, such as polymatroids~\cite{ed}, and more generally box-totally dual integer polyhedra~\cite{sc}. Box-integrality plays some role for polyhedra to have the Integer Carath\'eodory Property in~\cite{gire}. Binary clutters being $\frac{1}{k}$-box-integer for all $k\in\mathbb{Z}_{\scriptscriptstyle>0}$ are characterized in~\cite{gela}.

Actually, all these examples of box-integer polyhedra are principally box-integer. Our characterizations then yield new insights towards their properties.

\paragraph{Box-total dual integrality.}
Box-total dual integral (box-TDI) systems and polyhedra received a lot of attention from the combinatorial optimization community around the 80s. These systems yield strong combinatorial min-max relations with a geometric interpretation. A renewed interest appeared in the last decade and since then many deep results appeared involving such systems. 
The famous MaxFlow-MinCut theorem of Ford and Fulkerson~\cite{fofu} is a typical example of min-max relation implied by the box-TDIness of a system.
Other examples of fundamental box-TDI systems appear for polymatroids and for systems with a totally unimodular matrix of constraints.

A linear system $Ax\leq b$ is {\em totally dual integral} ({\em TDI\/}) if the maximum in the linear programming duality equation
$\max \{w^\top x:Ax\leq b\} = \min \{b^\top y:A^\top y = w, \, y \geq \mathbf 0\}$
has an integer optimal solution for all integer vectors $w$ for which the optimum is finite. Every polyhedron can be described by a TDI system~\cite[Theorem 22.6]{sc}. Moreover, the right hand side of such a TDI system can be chosen integer if and only if the polyhedron is integer~\cite{edgi}. A linear system $Ax\leq b$ is a {\em box-TDI system} if $Ax\leq b$, $\ell \le x \le u$ is TDI for each pair of rational vectors $\ell$ and $u$. In other words, $Ax\leq b$ is box-TDI if 
\begin{equation}\label{eq:boxTDI}
\min \{b^\top y + u^\top r - \ell^\top s :A^\top y + r - s = w, \, y\ge \mathbf{0},\, r\!,s \geq \mathbf 0\}
\end{equation}
has an integer solution for all integer vectors $w$ and all rational vectors $\ell,u$ for which the optimum is finite. It is well-known that box-TDI systems are TDI~\cite[Theorem~22.7]{sc}.
General properties of such systems can be found in~\cite{co},~\cite[Chap.~5.20]{scbig} and \cite[Chap.~22.4]{sc}.
Though not every polyhedron can be described by a box-TDI system, the result of Cook~\cite{co} below proves that being box-TDI is a property of the polyhedron. 
Consequently, a polyhedron that can be described by a box-TDI system is called a {\em box-TDI polyhedron}.
\begin{theorem}[{Cook~\cite[Corollary~2.5]{co}}]\label{tdiboxtdisystem}
If a system is box-TDI, then any TDI system describing the same polyhedron is also box-TDI.
%Let $Ax\leq b$ be a box-TDI system. If $Mx\leq d$ is a TDI system such that $\{x\in\mathbb{R}^{n}:Mx\leq d\}=\{x\in\mathbb{R}^{n}:Ax\leq b\}$, then the system $Mx\leq d$ is box-TDI.
\end{theorem}

\medskip

Originally, box-TDI systems were closely related to totally unimodular matrices. Indeed, any system with a totally unimodular matrix of constraint is box-TDI. Actually, until recently, the vast majority of known box-TDI systems were defined by a totally unimodular matrix, see~\cite{scbig} for examples.
When the constraint matrix is not totally unimodular, proving that a given system is box-TDI can be quite a challenge: one has to prove its TDIness, and then to deal with the addition of box-contraints that perturb the combinatorial interpretation of the underlying min-max relation. Ding, Feng, and Zang prove in~\cite{difeza} that it is {\em NP}-hard to recognize box-TDI systems.
%Incidentally, Gerards and Seb{\H{o}} prove in~\cite{gese} that the total dual integrality of a system has consequences related to unimodularity for the matrix involved.

Based on an idea of Ding and Zang~\cite{diza}, Chen, Chen, and Zang provide in~\cite{chchza} a sufficient condition for some systems to be box-TDI, namely the ESP property. Thanks to its purely combinatorial nature, the ESP property is successfully used to characterize: box-Mengerian matroid ports in~\cite{chchza}, the box-TDIness of the matching polytope in~\cite{ditaza}, subclasses of box-perfect graphs in~\cite{dizazh}. Prior to the development of the ESP property, the main tool to prove box-TDIness was~\cite[Theorem~22.9]{sc} of Cook. Its pratical application turns out to be quite technical as one has to combine polyhedral and combinatorial considerations, such as in~\cite{chdiza2} where the box-TDIness of a system describing the $2$-edge-connected spanning subgraph polytope on series-parallel graphs is proved.
In~\cite{cogrla}, Cornaz, Grappe, and Lacroix prove that a number of standard systems are box-TDI if and only if the graph is series-parallel. 

\paragraph{Contributions.} 
%This paper introduces a new class of polyhedra called principally box-integer polyhedra---see Definition~\ref{def:pbi}---and is devoted to their study. 
Our results provide a framework within which the notions of equimodularity, principal box-integrality, and box-TDIness are all connected.
The point of view obtained from principally box-integer polyhedra unveils new properties and simplifies the approach.

We now state our main result. A {\em face-defining} matrix for a polyhedron is a full row rank matrix describing the affine hull of a face of the polyhedron---see Section~\ref{sss:facedef} for more details.

\begin{theorem}\label{cor:main}
For a polyhedron $P$, the following statements are equivalent.
\begin{enumerate}%[(1)]
	\item\label{polyiv} The polyhedron $P$ is principally box-integer.
	\item\label{polyi} The polyhedron $P$ is box-TDI.
	\item\label{polyii-} Every face-defining matrix of $P$ is equimodular.
	\item\label{polyii} Every face of $P$ has an equimodular face-defining matrix.
	\item\label{polyiii} Every face of $P$ has a totally unimodular face-defining matrix.
\end{enumerate}
\end{theorem}

Along our proof, we show that a cone $\{x:Ax\leq \mathbf{0}\}$ is box-TDI if and only if it is box-integer, and that these properties are passed on to its polar. 
We use this to derive a polar version of Theorem~\ref{cor:main}---see Corollary~\ref{cor:pol}.

\medskip

These new results allow to prove the box-TDIness of systems by making full use of Theorem~\ref{tdiboxtdisystem}: find a TDI system describing the polyhedron on the one hand, and, on the other hand, apply one of the characterizations of principally box-integer polyhedra to prove the box-TDIness of the polyhedron. In particular, when a TDI system that describes the polyhedron is already known, our characterizations allow to pick whichever system---TDI or not---describing the polyhedron, and to use algebraic tools to prove the ``box'' part. The drawback of our characterization is that it does not provide a box-TDI system describing the polyhedron. Nevertheless, one of our characterizations gives an easy way to disprove box-TDIness: it is enough to exhibit a face-defining matrix having two maximal nonzero determinants of different absolute values.
In particular, this provides a simple co-NP certificate for the box-TDIness of a polyhedron.%decide whether a given polyhedron is box-TDI.

We show how known results on box-TDI polyhedra are simple consequences of our characterizations---see Section~\ref{s:connections}. We also explain how our results are connected with Schrijver's sufficient condition~\cite[Theorem~5.35]{scbig} and Cook's characterization~\cite{co}, {\cite[Theorem~22.9]{sc}}.

\medskip

We illustrate the use of our characterizations on several examples---see Section~\ref{s:illustrations}. 
First, we explain the equivalence between the main result of Gerards and Laurent~\cite{gela} and that of Chen, Ding, and Zang~\cite{chdiza} about binary clutters.
As a second application, we disprove a conjecture of Ding, Zang, and Zhao~\cite{dizazh} about box-perfect graphs. 
Then, we discuss Gijswijt and Regts~\cite{gire}'s abstract class of polyhedra having the Integer Carath\'eodory Property and possible connections between full box-integrality and the integer decomposition property. 
Finally, we prove that the cone of conservative functions of a graph is box-TDI if and only if the graph is series-parallel and we provide a box-TDI system describing it.

\paragraph {Outline.} 
Section~\ref{s:def} contains standard definitions. 
In Section~\ref{s:mbi}, we study general properties of principally box-integer polyhedra.
Section~\ref{s:pbiem} shows how equimodularity and principal box-integrality are intertwined: each notion is characterized using the other one.
In Section~\ref{s:boxtdi}, we first prove that a polyhedron is box-TDI if and only if it is principally box-integer, and then discuss the connections between our characterizations and existing results about box-TDI polyhedra.
In Section~\ref{s:illustrations}, we illustrate the use of our characterizations on several examples.

%%%%%%%%%%%%%%%%%%%%%%%%%%%%%%%%%%%%%%%%%%%%%%%%%%%%%%%%%%%%%%%%%%%%%%%%%%%%%%%%%%%%%%%%%%%%%%%%%%%%%%%%%%%%%%%%%%%%%%%%%%%%%%%%%%%%%%%%
%%%%%%%%%%%%%%%%%%%%%%%%%%%%%%%%%%%%%%%%%%%%%%%%%%%%%%%%%%%%%%%%%%%%%%%%%%%%%%%%%%%%%%%%%%%%%%%%%%%%%%%%%%%%%%%%%%%%%%%%%%%%%%%%%%%%%%%%
%%%%%%%%%%%%%%%%%%%%%%%%%%%%%%%%%%%%%%%%%%%%%%%%%%%%%%%%%%%%%%%%%%%%%%%%%%%%%%%%%%%%%%%%%%%%%%%%%%%%%%%%%%%%%%%%%%%%%%%%%%%%%%%%%%%%%%%%
%%%%%%%%%%%%%%%%%%%%%%%%%%%%%%%%%%%%%%%%%%%%%%%%%%%%%%%%%%%%%%%%%%%%%%%%%%%%%%%%%%%%%%%%%%%%%%%%%%%%%%%%%%%%%%%%%%%%%%%%%%%%%%%%%%%%%%%%
%%%%%%%%%%%%%%%%%%%%%%%%%%%%%%%%%%%%%%%%%%%%%%%%%%%%%%%%%%%%%%%%%%%%%%%%%%%%%%%%%%%%%%%%%%%%%%%%%%%%%%%%%%%%%%%%%%%%%%%%%%%%%%%%%%%%%%%%
%%%%%%%%%%%%%%%%%%%%%%%%%%%%%%%%%%%%%%%%%%%%%%%%%%%%%%%%%%%%%%%%%%%%%%%%%%%%%%%%%%%%%%%%%%%%%%%%%%%%%%%%%%%%%%%%%%%%%%%%%%%%%%%%%%%%%%%%

\section{Definitions}\label{s:def}

%%%%%%%%%%%%%%%%%%%%%%%%%%%%%%%%%%%%%%%%%%%%%%%%%%%%%%%%%%%%%%%%%%%%%%%%%%%%%%%%%%%%%%%%%%%%%%%%%%%%%%%%%%%%%%%%%%%%%%%%%%%%%%%%%%%%%%%%
%%%%%%%%%%%%%%%%%%%%%%%%%%%%%%%%%%%%%%%%%%%%%%%%%%%%%%%%%%%%%%%%%%%%%%%%%%%%%%%%%%%%%%%%%%%%%%%%%%%%%%%%%%%%%%%%%%%%%%%%%%%%%%%%%%%%%%%%

\paragraph{Matrices.} 
Throughout the paper, all entries will be rational. The $i$th unit vector of $\mathbb{R}^n$ will be denoted by $\chi^i$. For $I\subseteq\{1,\dots,n\}$, let $\chi^I=\sum_{i\in I} \chi^i$. 
An element $A$ of $\mathbb{R}^{m\times n}$ will be thought of as a matrix with $m$ rows and $n$ columns, and an element $b$ of $\mathbb{R}^m$ as a column vector. 
When all their entries belong to $\mathbb{Z}$, we will call them {\em integer}. 
%We will denote by $*$ the entries which can be any rational. 
The row vectors of $A$ will be denoted by $a_i^\top$, the column vectors of $A$ by $A^i$.
%, and $A\setminus A^i$ will denote the matrix obtained from $A$ by deleting the $i$th column. 
%A submatrix of $A$ having $n$ columns is a {\em row-submatrix} of $A$. 
When $\rank(A)=m$, we say that $A$ has {\em full row rank}.
A matrix is {\em totally unimodular}, or {\em TU}, if the determinants of its square submatrices are equal to $-1$, $0$ or $1$.

\paragraph{Lattices.} The {\em lattice} generated by a set~$V$ of vectors of~$\mathbb{R}^n$ is the set of integer combinations of these vectors, and is denoted by $\lattice(V)=\{\sum_{v\in V} \lambda_v v: \lambda_v\in\mathbb{Z}\mbox{ for all $v\in V$}\}$. The lattice generated by the column vectors of a matrix~$A$ is denoted by $\lattice(A)$.

\paragraph{Polyhedra.} 
Given $A\in\mathbb{Q}^{m\times n}$ and $b\in\mathbb{Q}^{m}$, the set $P=\{x\in\mathbb{R}^n: Ax\leq b\}=\{x\in\mathbb{R}^n: a^\top_i x\leq b_i, i=1,\dots,m\}$ is a {\em polyhedron}. 
We will often simply write $P=\{x: Ax\leq b\}$. 
The matrix $A$ is the {\em constraint matrix} of $P$.
The {\em translate of $P$ by $w\in\mathbb{R}^n$} is $P+w=\{x+w: x\in P\}$.

A {\em face} of $P$ is a nonempty set obtained by imposing equality on some inequalities in the description of $P$, that is, a nonempty set of the form $F=\{x: a^\top_i x=b_i, i\in I\}\cap P$ where $I\subseteq \{1,\dots ,m\}$. 
A row $a^\top_i$ or an inequality $a^\top_i x \leq b_i$ with $F\subseteq\{x:a^\top_i x = b_i\}$ is {\em tight for~$F$}, and $A_F x\leq b_F$ will denote the inequalities from $Ax\leq b$ that are tight for $F$. 
The set of points contained in $F$ and in no face $F'\subset F$ forms the {\em relative interior of $F$}.
Let $\lin(F)=\{x:A_F x=\mathbf{0}\}$ and $\aff(F)=\{x:A_F x=b_F\}$. 
The dimension $\dim(F)$ of a face $F$ is the dimension of its affine hull $\aff(F)$.
A {\em facet} is a face that is inclusionwise maximal among all faces distinct from $P$.
A face is {\em minimal} if it contains no other face of $P$. Minimal faces are affine spaces. A minimal face of dimension $0$ is called a {\em vertex}. 
Note that a polyhedron is integer if and if each of its minimal faces contains an integer point.

\paragraph{Cones.} 
A {\em cone} is a polyhedron of the form $t+\{x: Ax\leq \mathbf{0}\}$ for some $t\in\mathbb{R}^n$ and some $A\in\mathbb{R}^{m\times n}$.
When $t=\mathbf{0}$, the cone $C=\{x: Ax\leq \mathbf{0}\}$ can also be described as the set of nonnegative combinations of a set of vectors $R\subseteq \mathbb{R}^n$, and we say that $C=\cone(R)$ is {\em generated by}~$R$. 

The {\em polar cone} of a cone $C=\{x: Ax\leq \mathbf{0}\}$ is the cone $C^*=\{x: z^\top x \leq 0 \text{ for all $z\in C$}\}$. Equivalently, $C^*$ is the cone generated by the columns of $A^\top$. Note that $C^{**}=C$.

Given a face $F$ of a polyhedron $P=\{x: Ax\leq b\}$, the {\em tangent cone associated to $F$} is the cone $C_F=\{x:A_F x\leq b_F\}$. When $F$ is a minimal face of $P$, its associated cone is a {\em minimal tangent cone of $P$}. The cone of $\mathbb{R}^n$ generated by the columns of $A_F^\top$ is the {\em normal cone associated to $F$}. Note that the normal cone associated to $F$ is the polar of $\{x:A_F x\leq \mathbf{0}\}$.

\medskip

For more details, we refer the reader to Schrijver's book~\cite{sc}.

\section{Generalities on Principally Box-Integer Polyhedra}\label{s:mbi}

%%%%%%%%%%%%%%%%%%%%%%%%%%%%%%%%%%%%%%%%%%%%%%%%%%%%%%%%%%%%%%%%%%%%%%%%%%%%%%%%%%%%%%%%%%%%%%%%%%%%%%%%%%%%%%%%%%%%%%%%%%%%%%%%%%%%%%%%

This section is devoted to the basic properties of box-integer and principally box-integer polyhedra. In particular, we study the behavior of these notions with respect to dilation and translation.

\subsection{Box-Integer Polyhedra}\label{ss:PBI}

%%%%%%%%%%%%%%%%%%%%%%%%%%%%%%%%%%%%%%%%%%%%%%%%%%%%%%%%%%%%%%%%%%%%%%%%%%%%%%%%%%%%%%%%%%%%%%%%%%%%%%%%%%%%%%%%%%%%%%%%%%%%%%%%%%%%%%%%

Recall that a polyhedron $P$ is box-integer if $P\cap \{\ell\leq x\leq u\}$ is integer for all $\ell,u\in\mathbb{Z}^n$.
In proofs, the following characterization will often be more practical than the definition.

\begin{lemma}\label{boxint} 
A polyhedron $P$ is box-integer if and only if for each face $F$ of $P$, $I\subseteq \{1,\dots,n\}$, and $p\in\mathbb{Z}^I$  such that $\aff(F)\cap\{x_i=p_i, i\in I\}$ is a singleton $v$, if $v$ belongs to $F$ then $v$ is integer.
\end{lemma}
\begin{proof} Let $P=\{x\in\mathbb R^n: Ax\leq b\}$. By definition, $P$ is box-integer if and only if every vertex of $P\cap \{\ell\leq x\leq u\}$ is integer, for all $\ell,u\in\mathbb{Z}^n$.
A point $v\in\mathbb R^n$ is a vertex of $P\cap \{\ell\leq x\leq u\}$ if and only  if $v$ belongs to $P\cap \{\ell\leq x\leq u\}$ and $v$ is the unique solution of a non singular system $a_jx=b_j,j\in J, x_i=p_i,i\in I$ where $p_i\in\{\ell_i,u_i\}$. 
Note that $F=\{x: a_jx=b_j,j\in J\}\cap P$ is a face of $P$ and that $\aff(F)\cap\{x_i=p_i, i\in I\}$ is nothing but the set of solutions of the latter system.
\end{proof}

Note that, if $I$ is such that the set $\aff(F)\cap\{x_i=p_i, i\in I\}$ is a singleton for some $p\in\mathbb{R}^I$, then this set is either empty or a singleton for all $p\in\mathbb R^I$.
If $I$ is moreover assumed inclusionwise minimal, then $\aff(F)\cap\{x_i=p_i, i\in I\}$ is a singleton for all $p\in\mathbb R^I$.

The following two results seem to be known in the literature, we provide a proof for the sake of completeness.

\begin{corollary}\label{bii} 
If a polyhedron $P$ is box-integer, then $P$ is integer.
\end{corollary}
\begin{proof}
Let $F$ be a minimal face of $P$. There exists an inclusionwise minimal set $I$ as above, hence setting $\{x_i=p_i,i\in I\}$ for some $p\in\mathbb{Z}^I$ yields a singleton in~$\aff(F)$. Since $\aff(F)=F$, this singleton is integer by Lemma~\ref{boxint}, and thus $F$ contains an integer point.
\end{proof}

\begin{corollary}\label{corboxint} 
Let $P$ be a polyhedron of $\mathbb{R}^n$. The following statements are equivalent.
\begin{enumerate}
	\item\label{corboxinti}  $P$ is box-integer.
	\item\label{corboxintii}  $P\cap \{x\geq \ell\}$ is integer for all $\ell\in\mathbb{Z}^n$.
	\item\label{corboxintiii}  $P\cap \{\ell\leq x\leq u\}$ is integer for all $\ell,u\in\mathbb{Z}\cup\{-\infty,+\infty\}^n$.
\end{enumerate}
\end{corollary}
\begin{proof}
Point~\ref{corboxintiii} immediately implies point~\ref{corboxintii}. Point~\ref{corboxintii} implies point~\ref{corboxinti} by Lemma~\ref{boxint}, as if $\aff(F)\cap\{x_i=p_i,i\in I\}$ is a singleton $v\in F$, then $v$ is a vertex of $P\cap\{x\geq \lfloor v\rfloor\}$. Point~\ref{corboxinti} implies point~\ref{corboxintiii} because if $P$ is box-integer, then for all $\ell,u\in\mathbb{Z}\cup\{-\infty,+\infty\}^n$, $P\cap \{\ell \leq x\leq u\}$ is box-integer---and hence integer by Corollary~\ref{bii}.
\end{proof}

The following lemma shows two operations which preserve box-integrality. The second one will be used in Section~\ref{s:boxtdi}.

\begin{lemma}\label{pboxint}
Let $P=\{x\in\mathbb R^n: Ax\leq b\}$ be a polyhedron.
\begin{enumerate}%[(i)]
	\item\label{ptildeint}  $P$ is box-integer if and only if $\widetilde{P} = \{y,z\in\mathbb{R}^n:A(y+z)\leq b\}$ is box-integer.
	\item\label{pmint}  $P$ is box-integer if and only if $P_{\pm} = \{y,z\in\mathbb{R}^n:A(y-z)\leq b,\, y,z\geq \mathbf{0}\}$ is box-integer.
\end{enumerate}
\end{lemma}

\begin{proof}
For the first direction of point~\ref{ptildeint}, if $\widetilde P$ is box-integer, then so is $P=\widetilde{P}\cap \{z=\mathbf{0}\}$. 
For the other direction, we use Lemma~\ref{boxint}.
Let $F$ be a face of $\widetilde P$, of affine space $\aff(F)=\{y,z\in\mathbb R^n:a_j(y+z)=b_j, j\in J\}$ and $p,q$ be integer vectors such that $S=\aff(F)\cap \{y_i=p_i, i\in I_y, z_i=q_i, i\in I_z\}$ is a singleton $(\bar y, \bar z)$ which belongs to $F$.
Let us show that $(\bar y,\bar z)$ is integer. By Lemma~\ref{boxint}, this implies  that $\widetilde P$ is box-integer.

We denote by $G$ the face of $P$ of affine space $\{x\in\mathbb R^n:a_jx=b_j, j\in J\}$. 
Then $\aff(G)\cap \{x_i=p_i+q_i,i\in I_y\cap I_z\}$  is the singleton $\bar x=\bar y+\bar z$. Indeed, if it contained an other point $\bar x'$, we could set $\bar y'_i=p_i,i\in I_y, \bar z'_i=q_i,i\in I_z$ and then build $(\bar y',\bar z')$ in $S$ such that $\bar y'+\bar z'=\bar x'\neq \bar y+\bar z$, a contradiction.
$P$ is box-integer and $\bar y+\bar z$ belongs to $P$, thus $\bar y+\bar z$ is integer by Lemma~\ref{boxint}. 
Since $S$ is a singleton, no $(\bar y+\chi_i, \bar z-\chi_i)$ belongs to $S$, and for all $i$, we have either $y_i=p_i$ or $z_i=q_i$.
By $p$, $q$, and $\bar y +\bar z$ being integer, $(\bar y,\bar z)$ is  integer.

For the first direction of point~\ref{pmint}, if $P$ is box-integer, then so is $P_\pm$ by point~\ref{ptildeint} and because $P_\pm$ is obtained from $\widetilde{P}\cap \{y\geq \mathbf{0}, z\leq \mathbf{0}\}$ by replacing $z$ by~$-z$.
Suppose now that $P_\pm$ is box-integer. For $t\in\mathbb{R}^n$, define $t_+=\max\{\mathbf{0},t\}$ and $t_-=\max\{\mathbf{0},-t\}$. For $\ell,u\in\mathbb{Z}^n$, we have $u=u_+-u_-$, $\ell = \ell_+-\ell-$, and $u_+,u_-,\ell_+,\ell-\geq \mathbf{0}$, hence $P\cap \{\ell\leq x\leq u\}$ is the projection onto $x=y-z$ of $P_\pm\cap \{\ell_+\leq y \leq u_+,-\ell_-\leq -z \leq -u_-\}$. The latter being integer, this implies the integrality of $P\cap \{\ell\leq x\leq u\}$. 

%$P\cap \{\ell\leq x\leq u\}$ is the set of $x$ such that $x=y-z$ for some $(y,z)\in P_\pm\cap \{\ell_+\leq y \leq u_+,-\ell_-\leq -z \leq -u_-\}$. The latter being integer, this implies the integrality of $P\cap \{\ell\leq x\leq u\}$. 

%Let $F$ be a face of $P_\pm$  of affine space $\aff(F)=\{x\in\mathbb R^n:a_jx=b_j, j\in J\}$ and $p$ be an integer vector such that $S=\aff(F)\cap \{x_i=p_i, i\in I\}$ is a singleton $\bar x$ in $F$. 
%Let us show that $\bar x$ is integer. By Lemma~\ref{boxint} this implies  that $P$ is box-integer. 
%
%Let $\bar x^+$ and $\bar x^-$ be respectively obtained from $\bar x$ and $-\bar x$ by setting to zero the negative coordinates. Note that $\bar x=\bar x^+-\bar x^-$.
%Define similarly $p^+$ and $p ^-$, respectively  from $p$ and $-p$.  
%Let $K$ be the set of all indices $k$ such that $\bar x^+_k=0$.
%Since $S$ is the singleton $\bar x$,  the set $\{y,z\in\mathbb R^n:a_j(y-z)=b_j, j\in J,y_k=0,k\in K,z_k=0,k\notin K\}\cap \{y_i=p ^+_i, z_i=p^-_i,i\in I\}$ is the singleton $(\bar x^+,\bar x^-)$.
%Finally, by $(\bar x^+,\bar x^-)\in P_\pm$ and by $P_\pm$ being box-integer, $(\bar x^+,\bar x^-)$ is integer. 
%Hence so is $\bar x=\bar x^+-\bar x^-$.
\end{proof}

\subsection{Dilations of Box-Integer Polyhedra}

%dilation ou dilation ???

In this section we investigate how the box-integrality of a polyhedron behaves with respect to dilation. As a preliminary, the following observation describes the behaviour of integrality with respect to dilation.

\begin{proposition}\label{kbi}
Let $P$ be a polyhedron. There exists $d\in\mathbb{Z}_{\scriptscriptstyle >0}$ such that $\{k\in\mathbb{Z}_{\scriptscriptstyle >0}:\mbox{$kP$ is integer}\}=d\mathbb{Z}_{\scriptscriptstyle >0}$.
\end{proposition}
\begin{proof} 
When $P$ has vertices, it is enough to choose $d$ as the smallest positive integer $d$ such $dv$ is integer for every vertex $v$ of $P$. To treat the general case, we prove that if $kP$ and $k' P$ are integer polyhedra, then $\gcd(k,k') P$ is an integer polyhedron too. Then, the smallest positive integer $k$ such that $kP$ is integer divides all the others, and as any dilation of an integer polyhedron is an integer polyhedron too, this proves the observation.
 
Let $P=\{x:Ax\leq b\}$, $i=\gcd(k,{k'})$, $\overline k=k/i$, $\overline {k'}={k'}/i$, and $F$ be a minimal face of $iP$. Since $F$ is a minimal face, $F$ is the affine space $F=\{x:A_F x=ib_F\}$. Note that $\overline k F$ and $\overline {k'} F$ are minimal faces, respectively of $kP$ and ${k'} P$, thus contain an integer point, respectively $x_k$ and $x_{k'}$. By Bézout's lemma, there exist $\lambda$ and $\mu$ in $\mathbb{Z}$ such that $\lambda k+\mu {k'}=i$. Then $A_F(\lambda x_k+\mu x_{k'})=ib_F$, hence $F$ contains an integer point. Therefore, $\gcd(k,{k'}) P$ is an integer polyhedron.
\end{proof}

One of the arguments in the previous proof is the fact that the dilations of an integer polyhedron are also integer polyhedra. This does not hold for box-integrality, intuitively because any $0/1$ polytope is box-integer, though its dilations have no reasons to be. 
Actually, an example of box-integer polyhedron having non box-integer dilations will be provided at the end of this section. 
For now we prove the following lemma in order to determine, given a polyhedron~$P$, the structure of the set of positive integers $k$ such that $kP$ is box-integer.

\begin{lemma}\label{dil-notbi} Let $P$ be a polyhedron and $k\in\mathbb Z_{\scriptscriptstyle>0}$ such that $kP$ is integer but not box-integer. Then, no dilation $k' P$ with $k' \geq k$ is box-integer.
\end{lemma}

\begin{proof}  
Let $k'\geq k$. Assume $k'P$ integer, as otherwise $k'P$ would not be box-integer.
By Lemma~\ref{boxint}, there exist a face $F$ of $kP$ and an integer vector $p$ such that $\aff(F)\cap \{x_i=p_i, i\in I\}$ is a noninteger singleton $v\in F$. 
By Proposition~\ref{kbi}, $kP$ and $k'P$ are both dilations of an integer polyhedron $dP$. 
In particular, there exists an integer point $z$ in $F$ such that $z'=\frac{k'}{k}z$ is an integer point contained in the face $F'=\frac{k'}{k} F$ of $k'P$.
Since $k'\geq k$, we have $F-z \subseteq F'-z'$, thus $v'=(z'-z)+v$ is in~$F'$. 
Moreover, $\aff(F')\cap \{x_i=(z'_i-z_i)+p_i, i\in I\}$ is the singleton $v'$ of $F'$, which is not integer, hence $k'P$ is not box-integer by Lemma~\ref{boxint}.
\end{proof}

A polyhedron $P$ is {\em fully box-integer} if $kP$ is box-integer for all $k\in\mathbb{Z}_{\scriptscriptstyle >0}$. In other words, $P$ is fully box-integer if and only if $P$ is principally box-integer and integer.

\begin{proposition}\label{pbi-defs} For a polyhedron $P$, the following statements are equivalent.
\begin{enumerate}
\item\label{pbi-def} $P$ is principally box-integer.
\item\label{pbi-ideal} There exists $d\in\mathbb{Z}_{\scriptscriptstyle >0}$ such that $\{k\in\mathbb{Z}_{\scriptscriptstyle >0}:\mbox{$kP$ is box-integer}\}=d\mathbb Z_{\scriptscriptstyle>0}$.
\item \label{pbi-fbi} $P$ has a fully box-integer dilation.
\end{enumerate}
\end{proposition}
\begin{proof}  
The definition of principal box-integrality and Proposition~\ref{kbi} give \eqref{pbi-def}$\Rightarrow$\eqref{pbi-ideal}. To get \eqref{pbi-ideal}$\Rightarrow$\eqref{pbi-fbi}, just note that $dP$ is a fully box-integer polyhedron. To prove \eqref{pbi-fbi}$\Rightarrow$\eqref{pbi-def}, suppose that $P$ is not principally box-integer, that is, there exists a positive integer $k$ such that $kP$ is integer but not box-integer. By Lemma~\ref{dil-notbi}, this is not compatible with the existence of a fully box-integer dilation of $P$.
\end{proof}

We mention that relaxing $k\in\mathbb{Z}_{\scriptscriptstyle >0}$ to $k\in\mathbb{Z}$ in Definition~\ref{def:pbi} yields an equivalent definition. Then, the set arising in point~\ref{pbi-ideal} of Proposition~\ref{pbi-defs} is $d\mathbb{Z}$, which is a principal ideal of $\mathbb{Z}$. This explains why we called these polyhedra principally box-integer. 
The next proposition shows what can happen when a polyhedron is not principally box-integer.

\begin{proposition}\label{3cases} For a polyhedron $P$, exactly one of the following situations holds.
\begin{enumerate}[i.]
\item\label{3cases1} $P$ is principally box-integer.
\item\label{3cases2} No dilation of $P$ is a box-integer polyhedron.
\item\label{3cases3} There exist $d,q\in\mathbb Z_{\scriptscriptstyle >0}$ such that $kP$ is box-integer if and only if $k\in\{d,2d,\dots,qd\}$. 
\end{enumerate}
\end{proposition}
% 2 est le cas particulier de 3 où i=0 ?

\begin{proof}  
If $P$ has a  box-integer dilation but is not principally box-integer, then there is a smallest $q$ in $\mathbb{Z}_{\scriptscriptstyle >0}$ such that $(q+1) P$ is a polyhedron which is  integer but not box-integer.
By Lemma~\ref{dil-notbi}, no $kP$ with $k>q$ is box-integer.
Now, if $d$ is chosen as in Proposition~\ref{kbi}, the minimality of $q$ gives $\{k\in\mathbb{Z}_{\scriptscriptstyle >0}:\mbox{$kP$ is box-integer}\}=\{d,2d,\dots,qd\}$.
\end{proof}

Note that the following property, which holds for integrality, also holds for box-integrality: if $kP$ and $k' P$ are box-integer polyhedra, then so is $gcd(k,k')P$.

\begin{remark} 
Though we only considered dilations with positive integer coefficients, all these results can readily be adapted to dilations with rational coefficients.
\end{remark}

We conclude this section with an example of polyhedron whose box-integrality is not preserved by dilation.
As $P=\convexhull\left(\mathbf{0}, (1,1,0,0,0), (1,0,1,0,0), (1,0,0,1,0), (1,1,1,1,1)\right)$ is a $0/1$ polytope, it is box-integer. However, it can be checked that $(2, 1, 1, 1, 1/2)$ is a fractional vertex of $2P\cap\{x_2=x_3=x_4=1\}$. In particular, $P$ illustrates point~\ref{3cases3} of Proposition~\ref{3cases}.

%Dessins ?
% exemple du cas 1, exemple de pbi et fbi
% intuitions du lemme 2 
% peut-on donner une intuition de pourquoi la dilation d'un box-entier ne l'est pas forcément

\subsection{Translations of Principally Box-Integer Polyhedra}

Box-integrality is clearly preserved by integer translation. So are principal and full box-integrality.

\begin{obs} \label{obs:intrans}
Box-integrality, principal box-integrality and full box-integrality are all preserved by integer translation.
\end{obs}

\begin{proof}
The translate $Q=t+P$ of a box-integer polyhedron $P$ by $t$ in $\mathbb{Z}^n$ is also box-integer because $Q\cap\{\ell\leq x\leq u\}=t+(P\cap\{\ell-t\leq x \leq u-t\})$ for all $\ell,u\in\mathbb{Z}^n$. Moreover, since $kQ=kt+kP$ and $kt\in\mathbb{Z}^n$ for all $k\in\mathbb Z_{\scriptscriptstyle >0}$, principal box-integrality and full box-integrality are also preserved by integer translation.
\end{proof}

Cones will play an important role in the next sections. One of the reasons is that, up to translation, every dilation of a cone is the  cone itself. Since box-integrality is preserved by integer translation, this has the following consequences.

\begin{obs}\label{biconesprop} 
Let $C=\{x:Ax\leq \mathbf{0}\}$ be a cone of $\mathbb{R}^n$ and $D=t+C$ for some $t\in\mathbb{Q}^n$.
\begin{enumerate}
	\item\label{conesi} For $C$, the three properties of being box-integer, fully box-integer, or principally box-integer are equivalent.
%	\item\label{conesii} If $t\in \mathbb{Z}^n$, then $D$ is box-integer if and only if $C$ is box-integer.
	\item\label{conesiii} $D$ is fully box-integer if and only if it is box-integer.
	\item\label{conesiv} $D$ is principally box-integer if and only if $C$ is box-integer.
\end{enumerate}
\end{obs}

\begin{proof}
The fact that $kC=C$ for all $k\in\mathbb{Z}_{\scriptscriptstyle >0}$ proves point~\ref{conesi}. 
When $D$ is box-integer, its minimal face contains an integer point, hence $t$ can be chosen integer. 
Since $kD=(k-1)t+D$ for all $k\in\mathbb Z_{\scriptscriptstyle>0}$, and since integer translation preserves box-integrality, point~\ref{conesiii} follows.
When $t\in\mathbb{Q}^n$, take $k$ large enough such that $kt$ is integer. Now, $kD=kt+C$ is a fully box-integer dilation of $D$ if and only if $C$ is box-integer, which proves point~\ref{conesiv}.
\end{proof}

%%%%%%%%%%%%%%%%%%%%%%%%%%%%%%%%%%%%%%%%%%%%%%%%%%%%%%%%%%%%%%%%%%%%%%%%%%%%%%%%%%%%%%%%%%%%%%%%%%%%%%%%%%%%%%%%%%%%%%%%%%%%%%%%%%%%%%%%%%%%%%%%%%%%
%%%%%%%%%%%%%%%%%%%%%%%%%%%%%%%%%%%%%%%%%%%%%%%%%%%%%%%%%%%%%%%%%%%%%%%%%%%%%%%%%%%%%%%%%%%%%%%%%%%%%%%%%%%%%%%%%%%%%%%%%%%%%%%%%%%%%%%%%%%%%%%%%%%%
%%%%%%%%%%%%%%%%%%%%%%%%%%%%%%%%%%%%%%%%%%%%%%%%%%%%%%%%%%%%%%%%%%%%%%%%%%%%%%%%%%%%%%%%%%%%%%%%%%%%%%%%%%%%%%%%%%%%%%%%%%%%%%%%%%%%%%%%%%%%%%%%%%%%

\section{Principally Box-Integer Polyhedra and Equimodular Matrices}\label{s:pbiem}

%%%%%%%%%%%%%%%%%%%%%%%%%%%%%%%%%%%%%%%%%%%%%%%%%%%%%%%%%%%%%%%%%%%%%%%%%%%%%%%%%%%%%%%%%%%%%%%%%%%%%%%%%%%%%%%%%%%%%%%%%%%%%%%%%%%%%%%%%%%%%%%%%%%%
%%%%%%%%%%%%%%%%%%%%%%%%%%%%%%%%%%%%%%%%%%%%%%%%%%%%%%%%%%%%%%%%%%%%%%%%%%%%%%%%%%%%%%%%%%%%%%%%%%%%%%%%%%%%%%%%%%%%%%%%%%%%%%%%%%%%%%%%%%%%%%%%%%%%

In this section, we show how equimodularity and principal box-integrality are intertwined. First, we characterize equimodular matrices using principal box-integrality. Then, principally box-integer polyhedra are characterized by the equimodularity of a family of matrices.

\subsection{Characterizations of Equimodular Matrices}\label{ss:charactEM}

%%%%%%%%%%%%%%%%%%%%%%%%%%%%%%%%%%%%%%%%%%%%%%%%%%%%%%%%%%%%%%%%%%%%%%%%%%%%%%%%%%%%%%%%%%%%%%%%%%%%%%%%%%%%%%%%%%%%%%%%%%%%%%%%%%%%%%%%%%%%%%%%%%%%

%\subsection{Equimodular Matrices and Faces}\label{ss:EM}

%%%%%%%%%%%%%%%%%%%%%%%%%%%%%%%%%%%%%%%%%%%%%%%%%%%%%%%%%%%%%%%%%%%%%%%%%%%%%%%%%%%%%%%%%%%%%%%%%%%%%%%%%%%%%%%%%%%%%%%%%%%%%%%%%%%%%%%%

%Special matrices play a important role in our results, and we call them equimodular. 
%Recall that a $k\times n$ matrix is equimodular if its $k\times k$ nonzero determinants all have the same absolute value.
%These matrices are precisely the ones involved in~\cite[Theorem~19.5]{sc}: they are the matrices which are totally unimodular up to certain changes of basis.

In this section, we extend to equimodular matrices two classical results about unimodular matrices.
We first state the results of Heller~\cite{he} about unimodular sets in terms of equimodular matrices---see also~\cite[Theorem~19.5]{sc}.

\begin{theorem}[Heller~\cite{he}]\label{equimod}
 For a  full row rank $r\times n$ matrix $A$, the following statements are equivalent.
\begin{enumerate}
\item\label{equ} $A$ is equimodular.
\item\label{equlattice} For each nonsingular $r\times r$ submatrix $D$ of $A$, $\lattice(D)=\lattice(A)$.
\item\label{equint} For each nonsingular $r\times r$ submatrix $D$ of $A$, $D^{-1}A$ is integer.
\item\label{equ01} For each nonsingular $r\times r$ submatrix $D$ of $A$, $D^{-1}A$ is in $\{0,\pm1\}^{r\times n}$.
\item\label{equtu} For each nonsingular $r\times r$ submatrix $D$ of $A$, $D^{-1}A$ is totally unimodular.
\item\label{equtuex} There exists a nonsingular $r\times r$ submatrix $D$ of $A$ such that $D^{-1}A$ is totally unimodular.
\end{enumerate}
\end{theorem}

Veinott and Dantzig~\cite{veda} proved that an integer $r\times n$ matrix $A$ of full row rank is unimodular if and only if the polyhedron $\{x:Ax=b, x\geq \mathbf{0}\}$ is integer for all $b\in\mathbb{Z}^r$. 
Observe that point~\ref{corboxintii} of Corollary~\ref{corboxint} allows to reformulate their result as follows, since $\{x:Ax=kb\}\cap\{x\geq \ell\}=\ell+\{x:Ax=b', x\geq\mathbf{0}\}$, where $b'=kb+kA\ell\in\mathbb Z^r$.

\begin{theorem}[Veinott and Dantzig~\cite{veda}]\label{th:unimod}
Let $A$ be a full row rank matrix of $\mathbb{Z}^{r\times n}$.
Then, $A$ is unimodular if and only if $\{x:Ax=b\}$ is fully box-integer for all $b\in\mathbb{Z}^r$.
\end{theorem}

%\begin{proof}
%If $A$ is unimodular, then $\{x:Ax=kb\}\cap\{x\geq \ell\}=\ell+\{x:Ax=kb+kA\ell, x\geq\mathbf{0}\}$ is integer for all $\ell\in\mathbb{Z}^n$ by Veinott and Dantzig's result~\cite{veda}, which is as desired by point~\ref{corboxintii} of Corollary~\ref{corboxint}. Conversely, the box-integrality of $\{x:Ax=b\}$ implies the integrality of $\{x:Ax=b, x\geq \mathbf{0}\}$ for all $b\in\mathbb{Z}^r$, hence the unimodularity of $A$.
%\end{proof}

It turns out that this result can be extended to characterize equimodular matrices.

\begin{theorem}\label{thm:EMPBI}
Let $A$ be a full row rank matrix of $\mathbb{Q}^{r\times n}$.
Then, $A$ is equimodular if and only if $\{x:Ax=b\}$ is principally box-integer for all $b\in\mathbb{Q}^r$.
\end{theorem}

\begin{proof}
Suppose that $A$ is equimodular and let $b\in\mathbb{Q}^r$, $k\in\mathbb Z_{\scriptscriptstyle >0}$ be such that $H=\{x: Ax=kb\}$ is integer. Then $b'=kb$ belongs to $\lattice(A)$. 
Let $D$ be a nonsingular $r\times r$ submatrix~$D$ of~$A$. 
By point~\ref{equlattice} of Theorem~\ref{equimod}, we have $\lattice(D)=\lattice(A)$, hence $D^{-1}b'$ is in $\mathbb{Z}^r$.
Since $A$ has full row rank, by point~\ref{equtu} of Theorem~\ref{equimod}, $D^{-1}A$ is unimodular. 
By Theorem~\ref{th:unimod}, we get that $\{x: D^{-1}Ax=D^{-1}b'\}$ is fully box-integer. In particular, $H$ is box-integer.

Conversely, suppose that $A$ is not equimodular. Then, possibly reordering the columns, we may assume that the first $r$ columns of $A$ are linearly independent, and, by point~\ref{equint} of Theorem~\ref{equimod}, that the $r+1$th column $A^{r+1}$ of $A$ is a noninteger combination of those. Let $H=\{x:Ax=A^{r+1}\}$. Then, $\{x:Ax=A^{r+1}\}\cap\{x_j=0,j\geq r+1\}$ has no integer solution, hence $H$ is not box-integer. However, $H$ is integer as it contains $\chi^{r+1}$ as an integer point. Therefore, $H$ is not principally box-integer.
\end{proof}

Veinott and Dantzig~\cite{veda} devised Theorem~\ref{th:unimod} in order to get a simpler proof of a characterization of totally unimodular matrices due to Hoffman and Kruskal~\cite{hokr}. 
This characterization states that an integer matrix $A$ is totally unimodular if and only if $\{x:Ax\leq b\}$ is box-integer for all $b\in\mathbb{Z}^m$. In our context, this can be reformulated as follows.

\begin{theorem}[Hoffman and Kruskal~\cite{hokr}]\label{thm:hokr}
A matrix $A$ of $\mathbb{Z}^{m\times n}$ is totally unimodular if and only if $\{x:Ax\leq b\}$ is fully box-integer for all $b\in\mathbb{Z}^m$.
\end{theorem}

An equivalent definition of total unimodularity is to ask for every set of linearly independent rows to be unimodular. In this light, it is natural to define {\em totally equimodular matrices} as those for which all sets of linearly independent rows form an equimodular matrix. Theorem~\ref{thm:hokr} then extends to totally equimodular matrices as follows.

\begin{theorem}\label{TEMPBI}
A matrix $A$ of $\mathbb{Q}^{m\times n}$ is totally equimodular if and only if $\{x:Ax\leq b\}$ is principally box-integer for all $b\in\mathbb{Q}^m$.
\end{theorem}

\begin{proof}
Suppose $A$ totally equimodular and $b\in\mathbb{Q}^m$, and let us prove that $P=\{x:Ax\leq b\}$ is principally box-integer. Let ${k}\in\mathbb{Z}_{\scriptscriptstyle >0}$ be such that ${k} P$ is an integer polyhedron, and let us prove that ${k} P$ is box-integer. 
Let $F$ be a face of $kP$  and $p$ be an integer vector such that $\aff(F)\cap \{x_i=p_i, i\in I\}$ is a singleton $\bar x$ in $F$. 
By Lemma~\ref{boxint}, it remains to show that $\bar x$ is integer.
There exists a full row rank subset $L$ of rows of $A$ such that $\aff(F)=\{x: A_Lx=kb_L\}$. 
By $A$ being totally equimodular, $A_L$ is equimodular. By Theorem~\ref{thm:EMPBI}, $\aff(F)$ is principally box-integer. 
Now, $kP$ being integer, so is $\aff(F)$. 
Hence, $\aff(F)$ is box-integer and $\bar x$ is integer.

Suppose now that $A$ is not equimodular, that is, there exists a full row rank submatrix $A_L$ of size $r\times n$ of $A$ which is not equimodular. 
Then, we may assume that the first $r$ columns of $A_L$ are linearly independent, and that the $r+1$th column  of $A_L$ is a noninteger combination of those.
Let $\bar x$ be the unique solution of $A_L x =0, x_{r+1}=-1, x_j= 0,j>r+1$.
Then, $\bar x\notin \mathbb{Z}^n$. 
Define $b_L=\mathbf{0}$ and $b_j=1$ if $j\notin L$, and let us show that $P=\{x:Ax\leq b\}$ is not principally box-integer. 
There exists ${k}\in \mathbb{Z}_{\scriptscriptstyle >0}$ large enough such that $\bar x\in{k} P$, and such that ${k} P$ is integer. Then, ${k} P\cap\{x_{r+1}=-1, x_j= 0,j>r+1\}$ contains $\bar x$ as a vertex because $\bar x$ satisfies to equality $n$ linearly independent inequalities. Therefore, ${k} P$ is not box-integer.
\end{proof}

Since deciding whether a given matrix is totally unimodular can be done in polynomial time, see {\em e.g.} \cite[Chapter~20]{sc}, point~\ref{equtu} of Theorem~\ref{equimod} implies that deciding whether a given matrix is equimodular can be done in polynomial time.
For totally equimodular matrices, the situation is less clear. 
Nevertheless, each row of such a matrix is equimodular, and thus nonzero coefficients of a given row all have the same absolute value. 
As scaling each row preserves equimodularity, each equimodular matrix has a $0$, $\pm1$ representant, hence might have some combinatorial interpretation.
Since totally equimodular matrices generalize totally unimodular matrices and since the associated polyhedra have nice properties (see also Corolary~\ref{TEBTDI}), they are interesting by themselves and their structural properties deserve to be studied. We put forward the following problem as a first step in this direction.

\begin{open}
Can totally equimodular matrices be recognized in polynomial time?
\end{open}

\begin{remark}
The full row rank hypothesis made throughout this section is convenient, but not really necessary, provided the notions of unimodularity and equimodularity are correctly extended.
Hoffman and Kruskal~\cite{hokr} extend the notion of unimodularity to not necessarily full row rank matrices, and Theorem~\ref{th:unimod} still holds for those matrices~\cite[Page~301]{sc}.
The correct extension of equimodularity to general matrices is to require, for a matrix $A$ of rank $r$, that each set of $r$ linearly independent rows of $A$ forms an equimodular matrix. Properties of such matrices are studied in~\cite{he}.
We mention that none of the definitions and results of this paper are affected if these extended definitions are adopted and the full row rank hypothesis removed.
\end{remark}

\subsection{Affine Spaces and Face-Defining Matrices}\label{sss:facedef}

Affine spaces being special cases of cones, by point~\ref{conesiv} Observation~\ref{biconesprop}, $\{x:Ax=b\}$ is principally box-integer for all $b$ if and only if $\{x: Ax=\mathbf{0}\}$ is fully box-integer. In particular, one can drop the quantification over all $b\in\mathbb{Q}^n$ from Theorem~\ref{thm:EMPBI} as follows.

\begin{corollary}\label{thm:affbi}
Let $A$ be a full row rank matrix of $\mathbb{Q}^{r\times n}$ and $b\in\mathbb{Q}^n$.
Then, $A$ is equimodular if and only if the affine space $\{x: Ax=b\}$ is principally box-integer.
\end{corollary}

An affine space $\{x: Ax=b\}$ being integer if and only if $b$ belongs to $\lattice(A)$, the previous result has the following immediate consequence.

\begin{corollary}\label{thm:affbibis}
Let $A$ be a full row rank matrix of $\mathbb{Q}^{r\times n}$ and $b\in\mathbb{Q}^n$.
The affine space $\{x: Ax=b\}$ is fully box-integer if and only if $A$ is equimodular and $b\in\lattice(A)$.
\end{corollary}

Corollary~\ref{thm:affbi} yields a correspondence between equimodular matrices and principally box-integer affine spaces. We shall see in the next section that this correspondence, when applied to the faces of a polyhedron, provides a characterization of principally box-integer polyhedra. This motivates the following definition.

\paragraph{Face-defining matrices.}
Let $P=\{x: Ax\leq b\}$ be a polyhedron of $\mathbb{R}^n$ and $F$ be a face of $P$.
A full row rank matrix $M$ such that $\aff(F)$ can be written $\{x:Mx=d\}$ for some $d$ is {\em face-defining} for~$F$. 
Such matrices are called {\em face-defining matrices of $P$}.
Note that face-defining matrices need not correspond to valid inequalities for the polyhedron. 
%The face-defining matrices of a face $F$ have full row rank, and in particular they all have the same number of rows. 
%The only face defining matrix for the empty face is a single row only composed of zeros.
A face-defining matrix for a facet of $P$ is called {\em facet-defining}. 

\medskip

%The face-defining matrices of an affine space being obtained one from another by linear combinations of their rows, the multilinearity of the determinant and Theorem~\ref{equimod} imply the following. 
%This observation is probably well known, yet we found no reference for it hence we provide a proof.
Affine spaces are polyhedra whose only face is themselves. The following observation characterizes their principal box-integrality in terms of face-defining matrices.

\begin{obs}\label{Fequimod} For an affine space $H$, the following statements are equivalent.
\begin{enumerate}
\item\label{Fpbi} $H$ is principally box-integer.
\item\label{Fequex} $H$ has an equimodular face-defining matrix.
\item\label{Fequall} Every face-defining matrix of $H$ is equimodular.
%Avant il y avait ça: Every matrix which is face-defining for $H$ is equimodular.
\item\label{Fequtuex} $H$ has a totally unimodular face-defining matrix.
\end{enumerate}
\end{obs}
\begin{proof}
The equivalence between points~\ref{Fpbi}, \ref{Fequex}, and \ref{Fequall} follows from Corollary~\ref{thm:affbi}.
The equivalence between points~\ref{Fequex} and \ref{Fequtuex} follows from point~\ref{equtu} of Theorem~\ref{equimod}, because if $A\in\mathbb{Q}^{r\times n}$ is face-defining for $H$, then so is $D^{-1}A$ for each nonsingular $r\times r$ submatrix $D$ of $A$.
 %Let $M\in\mathbb{Q}^{r\times n}$ be face-defining for $H$. Then, the face-defining matrices for $H$ are exactly the matrices $DM$, for all nonsingular $r\times r$ matrices $D$.
%By property of the determinant, $M$ is equimodular if and only if $DM$ is equimodular.
%This gives the equivalence between point~\ref{Fequex} and point~\ref{Fequall}.
%The equivalence between point~\ref{Fequex} and point~\ref{Fequtuex} then follows from point~\ref{equtu} of Theorem~\ref{equimod} because, since $M$ has full row rank, one can chose for $D$ to be the inverse of some $r\times r$ nonsingular submatrix of $M$.
%PLus clair avec: Multiplying Mx=b = D^-1Mx= D^-1b and D^-1M is TU ?
\end{proof}

Note that, when $P$ is full-dimensional, facet-defining matrices are composed of a single row and are uniquely determined, up to multiplying by a scalar. 
In general, the number of rows of a face-defining matrix for a face $F$ is $n-\dim(F)$. 
More precisely, the following immediate observation characterizes face-defining matrices.

\begin{obs}\label{face-def_C}
A full row rank matrix $M\in\mathbb{Q}^{k\times n}$ is face-defining for a face~$F$ of a polyhedron $P\subseteq \mathbb{R}^n$ if and only if there exist a vector $d\in\mathbb{Q}^k$ and a family $\mathcal{H}\subseteq F\cap \{x: Mx=d\}$ of $\dim(F)+1$ affinely independent points such that $|\mathcal{H}|+k = n+1$. 
\end{obs}

\subsection{Characterizations of Principally Box-Integer Polyhedra}

In this section, we provide several characterizations of principally box-integer polyhedra, the starting point being the following lemma.

\begin{lemma}\label{lm:pbiaff}
A polyhedron $P$ is principally box-integer if and only if $\aff(F)$ is principally box-integer for each face $F$ of $P$.
\end{lemma}
\begin{proof} 
Let $P$ be a polyhedron such that the affine spaces generated by its faces are all principally box-integer. Then, when $k\in\mathbb Z_{\scriptscriptstyle >0}$ is such that $kP$ is integer, all the affine spaces generated by the faces of $kP$ are box-integer.
Therefore, by Lemma~\ref{boxint}, if $F$ is a face of such a $kP$ and $p$ is an integer vector such that $\aff(F)\cap \{x_i=p_i, i\in I\}$ is a singleton in $F$, then this singleton is integer. 
Then, by the other direction of Lemma~\ref{boxint}, $kP$ is box-integer, thus $P$ is principally box-integer.

Conversely, let $P$ be a principally box-integer polyhedron and $F$ be a face of $P$. If $F$ is a singleton, then $\aff(F)=F$ is a singleton, thus obviously principally box-integer.
Otherwise, let $t$ be a rational point in the relative interior of $F$, let $G=F-t$ and $Q=P-t$.
By point 3 of Observation~\ref{biconesprop}, it suffices to show that $\aff(G)$ is box-integer.
Let $p$ be an integer vector such that $\aff(G)\cap \{x_i=p_i, i\in I\}$ is a singleton $\bar x$ in $\aff(G)$.
By the hypothesis made on $t$, there exists $k\in\mathbb Z_{\scriptscriptstyle >0}$ such that $\bar x\in kQ$.
Moreover, such a $k$ can be chosen so that $kt$ is integer and $kP$ is an integer  polyhedron.
Since $P$ is principally box-integer, $kP$ is box-integer and so is $kQ=kP-kt$ by Observation~\ref{obs:intrans}. 
Applying Lemma~\ref{boxint} to the face $kG$ of $kQ$ yields $\bar x$ integer.
By applying the other direction of Lemma~\ref{boxint} to the unique face $\aff(G)$ of $\aff(G)$, we obtain that $\aff(G)$ is box-integer.
\end{proof}

\begin{theorem}\label{pbiEM} For a polyhedron $P$, the following statements are equivalent.
\begin{enumerate}
\item\label{pbiEM1} The polyhedron $P$ is principally box-integer.
\item\label{pbiEM3} Every minimal tangent cone of $P$ is principally box-integer.
\item\label{pbiEM2} Every face of $P$ has an equimodular face-defining matrix.
\end{enumerate}
\end{theorem}
\begin{proof}
Each face of $P$ is contained in a face of some minimal tangent cone of $P$ having the same affine hull. 
Conversely, each face of a minimal tangent cone of $P$ contains some face of $P$ having the same affine hull.
Therefore, Lemma~\ref{lm:pbiaff} gives the equivalence between point~\ref{pbiEM1} and point~\ref{pbiEM3}.
The equivalence between point~\ref{pbiEM1} and point~\ref{pbiEM2} is immediate by Corollary~\ref{thm:affbi} and Lemma~\ref{lm:pbiaff}. 
\end{proof}

The minimal faces of a polyhedron being affine spaces, Lemma~\ref{lm:pbiaff} has a fully box-integer counterpart. Moreover, by point~\ref{conesiii} of Observation~\ref{biconesprop}, so does the equivalence between point~\ref{pbiEM1} and point~\ref{pbiEM2} of Theorem~\ref{pbiEM}. This gives the following corollary.

\begin{corollary}
For a polyhedron $P$, the following statements are equivalent.
\begin{enumerate}
\item\label{cpbiEM1} The polyhedron $P$ is fully box-integer.
\item\label{cpbiEM3} Every minimal tangent cone of $P$ is box-integer.
\item\label{cpbiEM2} For each face $F$ of $P$, $\aff(F)$ is fully box-integer.
\end{enumerate}
\end{corollary}

%\begin{corollary}\label{loulou}
%A cone $C=\{x\in\mathbb{R}^n: Ax\geq \mathbf{0}\}$ is box-integer if and only if all its faces are equimodular.
%\end{corollary}

%%%%%%%%%%%%%%%%%%%%%%%%%%%%%%%%%%%%%%%%%%%%%%%%%%%%%%%%%%%%%%%%%%%%%%%%%%%%%%%%%%%%%%%%%%%%%%%%%%%%%%%%%%%%%%%%%%%%%%%%%%%%%%%%%%%%%%%%%%%%%%%%%%%%
%%%%%%%%%%%%%%%%%%%%%%%%%%%%%%%%%%%%%%%%%%%%%%%%%%%%%%%%%%%%%%%%%%%%%%%%%%%%%%%%%%%%%%%%%%%%%%%%%%%%%%%%%%%%%%%%%%%%%%%%%%%%%%%%%%%%%%%%%%%%%%%%%%%%
%%%%%%%%%%%%%%%%%%%%%%%%%%%%%%%%%%%%%%%%%%%%%%%%%%%%%%%%%%%%%%%%%%%%%%%%%%%%%%%%%%%%%%%%%%%%%%%%%%%%%%%%%%%%%%%%%%%%%%%%%%%%%%%%%%%%%%%%%%%%%%%%%%%%

\section{Box-Totally Dual Integral Polyhedra}\label{s:boxtdi}

%%%%%%%%%%%%%%%%%%%%%%%%%%%%%%%%%%%%%%%%%%%%%%%%%%%%%%%%%%%%%%%%%%%%%%%%%%%%%%%%%%%%%%%%%%%%%%%%%%%%%%%%%%%%%%%%%%%%%%%%%%%%%%%%%%%%%%%%

\subsection{New Characterizations of Box-TDI Polyhedra}\label{ss:boxtdi}

The main result of this section
is that the notions of principal box-integrality and box-TDIness coincide---see Theorem~\ref{main} below. Combined with Theorem~\ref{pbiEM}, this provides several new characterizations of box-TDI polyhedra.

\begin{theorem}\label{main}
A polyhedron is box-TDI if and only if it is principally box-integer.
\end{theorem}

\begin{proof} The proof relies on Lemmas~\ref{lem:polyBTDIminconesBTDI} and~\ref{boxtdicone}, which are proven below.

Lemma~\ref{lem:polyBTDIminconesBTDI} states that a polyhedron is box-TDI if and only if all its minimal tangent cones ares box-TDI. By Theorem~\ref{pbiEM}, a polyhedron is principally box-integer if and only if all its minimal tangent cones are principally box-integer. Hence it is enough to prove Theorem~\ref{main} for cones. 

Lemma~\ref{boxtdicone} states that a cone of the form $\{x:Ax\leq\mathbf 0\}$ is box-TDI if and only if it is box-integer. 
Then, by point~\ref{conesiv} of Observation~\ref{biconesprop}, and since box-TDIness is preserved under rational translation, a cone is box-TDI if and only if it is principally box-integer. 
\end{proof}
%%%%%%%%%%%%%%%%%%%%%%%%%%%%%%%%%%%%%%%%%%%%%%%%%%%%%%%%%%%%%%%%%%%%%%%%%%%%%%%%%%%%%%%%%%%%%%%%%%%%%%%%%%%%%%%%%%%%%%%%%%%%%%%%%%%%%%%%
%%%%%%%%%%%%%%%%%%%%%%%%%%%%%%%%%%%%%%%%%%%%%%%%%%%%%%%%%%%%%%%%%%%%%%%%%%%%%%%%%%%%%%%%%%%%%%%%%%%%%%%%%%%%%%%%%%%%%%%%%%%%%%%%%%%%%%%%

The following lemma seems somewhat implicitely known in the literature, but is not stated explicitely to the best of our knowledge. For the sake of completeness, we provide a proof which relies only on the definitions. It can also be shown using known characterizations of box-TDI polyhedra, such as the one by Cook~{\cite[Theorem~22.9]{sc}}.

\begin{lemma}\label{lem:polyBTDIminconesBTDI}
A polyhedron is box-TDI if and only if all its minimal tangent cones are.
\end{lemma}
\begin{proof}
Let $P=\{x:Ax\leq b\}$ be a polyhedron of $\mathbb{R}^n$ and $w\in\mathbb{Z}^n$. We will denote $(P_{\ell,u})=\max \{wx: Ax\leq b, \ell\leq x \leq u\}$ and $(P^F_{\ell,u})=\max \{wx: A_{I}x\leq b_I, \ell\leq x \leq u\}$ for a minimal face $F$ of $P$ where $I$ is the index set of the tight rows for $F$.

To prove the first direction, suppose that the system $Ax\leq b$ is box-TDI. Let $F$ be a minimal face of $P$, $v\in F$ and let $x^\star$ be an optimal solution of $(P^F_{\ell,u})$. Since $a_i v< b_i$ for all $i\notin I$, there exists $\lambda > 0$ such that $y^\star=v+\lambda(x^\star-v)$ belongs to $P$ and $a_i y^\star < b_i$ for all $i\notin I$. Let $\ell'=v+\lambda(\ell-v)$ and $u'=v+\lambda(u-v)$. Then, $y^\star$ is an optimal solution of $(P_{\ell',u'})$, as otherwise $x^\star$ would not be an optimal solution of $(P^F_{\ell,u})$. Let $(z^\star,r^\star,s^\star)$ be an integer optimal solution of the dual of $(P_{\ell',u'})$. By complementary slackness, denoting by $z_I^\star$ the vector obtained from $z^\star$ by deleting the coordinates not in $I$, without loss of generality we have $z^\star=(z_I^\star,\mathbf{0})$. 
Now, since $w^\top y^\star = b^\top z^\star +u'^\top r^\star -\ell'^\top s^\star$, one can check that $w^\top x^\star = b_I^\top z_I^\star +u^\top r^\star -\ell^\top s^\star$, by applying the definition of $y^\star$, $u'$ and $\ell'$, $b^\top z^\star=b_I^\top z_I^\star$, $w=A^\top z^\star+r^\star-s^\star$,
$A^\top z^\star=A_I^\top z_I^\star$, and $A_Iv=b_I$. Therefore, $(z_I^\star,r^\star,s^\star)$ is an integer optimal solution of the dual $\min \{b_I^\top z + u^\top r - \ell^\top s :A_I^\top z_I + r - s = w, \, z_I, r,s \geq \mathbf 0\}$ of $(P^F_{\ell,u})$.

For the other direction, let $H$ be the face of $P$ composed of all the optimal solutions of $(P_{\ell,u})=\max \{wx: Ax\leq b, \ell\leq x \leq u\}$ and let $F$ be a minimal face of $P$ contained in $H$ whose tight rows are indexed by $I$. Let $(z_I^\star,r^\star,s^\star)$ be an integer optimal solution of the dual of~$(P^F_{\ell,u})$. Then, one can check that extending $z_I^\star$ to a vector $z^\star=(z_I^\star,\mathbf{0})$ of $\mathbb{R}^m$ yields an integer optimal solution $(z^\star,r^\star,s^\star)$ of the dual of $(P_{\ell,u})$. 
\end{proof}

The following result reveals that cones behave nicely with respect to box-TDIness. 
It is already known that a box-TDI cone of the form $\{x: Ax\leq \mathbf{0}\}$ is box-integer~\cite[Equation~(5.82)]{scbig}.
Suprisingly, the converse holds and these properties are passed on to the polar.

\begin{lemma}\label{boxtdicone}
For a cone $C=\{x: Ax\leq \mathbf{0}\}$ of $\mathbb{R}^n$, the following statements are equivalent.
\begin{enumerate}%[(i)]
%\hspace{2.5cm}
\begin{minipage}{.22\textwidth}
	\item\label{cboxtdi} $C$ is box-TDI,
\end{minipage}
\begin{minipage}{.22\textwidth}
	\item\label{cboxint} $C$ is box-integer,
\end{minipage}
\begin{minipage}{.22\textwidth}
%\begin{enumerate}%[(i)]
	\item\label{c*boxtdi} $C^*$ is box-TDI,
\end{minipage}
\begin{minipage}{.22\textwidth}
	\item\label{c*boxint} $C^*$ is box-integer.
%\end{enumerate}
\end{minipage}
\end{enumerate}
\end{lemma}

\begin{proof} By \cite[Theorem 22.6(i)]{sc}, we can assume that $Ax\leq \mathbf{0}$ is a TDI system. 

Suppose that $C$ is box-TDI. By Theorem~\ref{tdiboxtdisystem}, the system $Ax\leq \mathbf{0}$ is box-TDI.
Hence, for all $\ell,u\in \mathbb{Z}^n$, the system $Ax\leq \mathbf{0}, \ell\leq x \leq u$ is TDI. As $\ell$ and $u$ are integer, this system defines an integer polyhedron by~\cite[Corollary 22.1c]{sc}.
Therefore, $C$ is box-integer, and we get \eqref{cboxtdi}$\Rightarrow$\eqref{cboxint}. This also gives \eqref{c*boxtdi}$\Rightarrow$\eqref{c*boxint}.

All that remains to prove is \eqref{c*boxint}$\Rightarrow$\eqref{cboxtdi}. Indeed, applying this implication to the cone $C^*$ and using that $C^{**} = C$ yields \eqref{cboxint}$\Rightarrow$\eqref{c*boxtdi}. 

Suppose that $C^*$ is box-integer and let us prove that the dual $(D)$ of the linear program $(P)$ below has an integer solution for all $w\in\mathbb{Z}^n$ and $\ell,u\in\mathbb{Q}^n$ such that the optimum is finite.
$$
(P)~~~
\begin{array}{ccccc}
\max & w^\top x & & & \\
& Ax & \leq & \mathbf{0} & \\
& x & \leq & u & \\
& - x & \leq & -\ell & \\
\end{array}
\hspace{2cm}
(D)~~~
\begin{array}{cccccccc}
\min & & & u^\top  r & - & \ell^\top  s & & \\
& A^\top z& +& r & - & s & = & w\\
& z & , & r & , & s & \geq & \mathbf{0}\\
&&&&&&&\\
\end{array}
$$

The projection of the set of $(z,r,s)$ satisfying the constraints of $(D)$ onto the variables $r$ and $s$ is the polyhedron $Q=\{r,s\geq \mathbf{0} : v^\top  (s-r+w)\leq 0, \text{ for all $v\in K$}\}$, where $K$ is the projection cone $K=\{v\in\mathbb{R}^n: v^\top  A^\top \leq \mathbf{0}\}$. That is $K=C$ and therefore $Q= (C^*-w)_\pm$.
Since integer translations of box-integer polyhedra are box-integer, $C^*-w$ is box-integer. Thus, by point~\ref{pmint} of Lemma~\ref{pboxint}, $Q$ is box-integer. Hence $Q$ is integer and there exists an integer solution $(\bar r,\bar s)$ maximizing $u^\top r-\ell^\top s$ over $Q$. Let $\bar w = w-\bar r+\bar s$. Since $(\bar r, \bar s)$ belongs to $Q$, there exists a feasible solution $\bar z$ of the dual of $\max\{ \bar w^\top x: Ax\leq \mathbf{0}\}$. Now, by $Ax\leq \mathbf{0}$ being TDI and by $\bar w$ being integer, such a $\bar z$ can be chosen integer. Then, $(\bar z, \bar r, \bar s)$ is an integer optimal solution of $(D)$.
 \end{proof}

%%%%%%%%%%%%%%%%%%%%%%%%%%%%%%%%%%%%%%%%%%%%%%%%%%%%%%%%%%%%%%%%%%%%%%%%%%%%%%%%%%%%%%%%%%%%%%%%%%%%%%%%%%%%%%%%%%%%%%%%%%%%%%%%%%%%%%%%

%%%%%%%%%%%%%%%%%%%%%%%%%%%%%%%%%%%%%%%%%%%%%%%%%%%%%%%%%%%%%%%%%%%%%%%%%%%%%%%%%%%%%%%%%%%%%%%%%%%%%%%%%%%%%%%%%%%%%%%%%%%%%%%%%%%%%%%%

We are now ready to prove our main result, Theorem~\ref{cor:main}.

\begin{proof}[Proof of Theorem~\ref{cor:main}]
Points~\ref{polyiv} and~\ref{polyi} are equivalent by Theorem~\ref{main}. Points~\ref{polyiv} and~\ref{polyii} are equivalent by the equivalence between points~\ref{pbiEM1} and~\ref{pbiEM2} of Theorem~\ref{pbiEM}. Finally, the equivalence between points~\ref{polyii-},~\ref{polyii}, and~\ref{polyiii} comes from Observation~\ref{Fequimod}. 
\end{proof}

We now apply polarity to derive additional characterizations of box-TDI polyhedra.

\begin{corollary}\label{cor:pol}
For a polyhedron $P$, the following statements are equivalent.
\begin{enumerate}
	\item\label{pol1} The polyhedron $P$ is box-TDI.
	\item\label{pol2} For every face $F$ of $P$, every basis of $\lin(F)$ is the transpose of an equimodular matrix.
	\item\label{pol3} For every face $F$ of $P$, some basis of $\lin(F)$ is the transpose of an equimodular matrix.
	\item\label{polylinF} For every face $F$ of $P$, some basis of $\lin(F)$ is a totally unimodular matrix.
\end{enumerate}
\end{corollary}
\begin{proof}
Let $F$ be a face of $P$. 
By Corollary~\ref{thm:affbi}, $F$ has an equimodular face-defining matrix if and only if $\aff(F)$ is principally box-integer. 
Equivalently, by Observation~\ref{biconesprop}, $\lin(F)$ is box-integer.
By Lemma~\ref{boxtdicone}, $\lin(F)$ is box-integer if and only $\lin(F)^*$ is.
By Corollary~\ref{thm:affbi}, $\lin(F)^*$ is box-integer if and only if $\lin(F)^*$ has an equimodular face-defining matrix $M$. 
Note that the columns of $M^\top$ form a basis of $\lin(F)$, therefore $F$ has an equimodular face-defining matrix if and only if some basis of $\lin(F)$ is the transpose of an equimodular matrix.

Since, by Theorem~\ref{cor:main}, the polyhedron $P$ is box-TDI if and only if each of its faces $F$ has an equimodular face-defining matrix, this proves the equivalence between points~\ref{pol1} and~\ref{pol3}. The equivalence with the two others points follows from Observation~\ref{Fequimod}.
\end{proof}

%By the properties of equimodular matrices and linear spaces, note that point~\ref{polylinF} of Corollary~\ref{cor:pol} can be expressed as the following two equivalent statements.
%A polyhedron $P$ is principally box-integer if and only if for each face $F$ of $P$, every basis of $\lin(F)$ forms a matrix whose transpose is equimodular. Equivalently, for each set $S=\{s_0,\dots,s_{\dim(F)}\}$ of affinely independent points of $F$, the matrix $[s_1-s_0,\dots,s_{\dim(F)}-s_0]^\top$ is equimodular.

Recall that a cone $C=\{x:Ax\leq \mathbf{0}\}$ can also be defined as $C=\cone(R)$ for some set $R$ of generators. Moreover, by Lemma~\ref{boxtdicone}, such a cone is box-TDI if and only if it is box-integer. Corollary~\ref{cor:pol} then allows to check whether such cones are box-integer by looking at their generators.

\begin{corollary}\label{conesrays}
A cone $C=\cone(R)$ is box-integer if and only if $S^\top$ is equimodular for each linearly independent subset $S$ of $R$ generating a face of $C$.
\end{corollary}
%\begin{proof}
%By Lemma~\ref{boxtdicone}, $C$ is box-integer if and only if $C^*$ is. By the equivalence between points~\ref{pbiEM1} and~\ref{pbiEM2} of Theorem~\ref{pbiEM}, this is equivalent to each face-defining matrix of $C^*$ being equimodular. Recall that $C^*=\{x\in\mathbb{R}^n: R^\top x\leq \mathbf{0}\}$. Moreover, a submatrix of $R^\top$ is face-defining for $C^*$ if and only if the corresponding vectors are linearly independent and generate a face of $C$. This gives the desired equivalence.
%\end{proof}

Consequently, the recognition of box-integer cones might have a different complexity status than the following related problems, which are all co-NP-complete: deciding whether a given polytope is integer~\cite{paya}, deciding whether a given system is TDI or box-TDI~\cite{difeza}, deciding whether a given conic system is TDI~\cite{pa}.

\begin{open}
What is the complexity of deciding whether a given cone is box-integer?
\end{open}

We mention that polarity preserves box-integrality only for cones of the form $\{x: Ax\leq \mathbf{0}\}$, and does not extend to polyhedra. For instance, the polyhedron $\convexhull\left((2,-1),(-2,-1),(0,1)\right)$ is fully box-integer, and its polar $\convexhull\left((1,1),(-1,1),(0,-1)\right)$ is integer but not box-integer. 

%%%%%%%%%%%%%%%%%%%%%%%%%%%%%%%%%%%%%%%%%%%%%%%%%%%%%%%%%%%%%%%%%%%%%%%%%%%%%%%%%%%%%%%%%%%%%%%%%%%%%%%%%%%%%%%%%%%%%%%%%%%%%%%%%%%%%%%%

\subsection{Connections with Existing Results}\label{s:connections}

In this section, we investigate the connections of our results with those from the literature about box-TDI polyhedra.
We first derive known results about box-TDI polyhedra from our characterizations. 
Then, we show how Cook's characterization~\cite[Theorem~22.9]{sc} is connected to ours. 
Finally, we discuss Schrijver's sufficient condition~\cite[Theorem~5.35]{scbig}.

\subsubsection{Consequences}\label{ss:csq}

Here, we review several known results about box-TDI polyhedra which can be derived from our results. The {\em dominant} of a polyhedron $P$ of $\mathbb{R}^n$ is $\dom(P)=P+\mathbb{R}^n_+$. 

\begin{csq}[{\cite[{Theorem~3.6}]{co} or~\cite[{Theorem 22.11}]{sc}}]
The dominant of a box-TDI polyhedron is box-TDI.
\end{csq}
\begin{proof}
The minimal tangent cones of the dominant of a polyhedron $P$ being the dominants of the minimal tangent cones of $P$, by Lemma~\ref{lem:polyBTDIminconesBTDI} it is enough to prove the result for a cone $C$. Moreover, 
as box-TDIness is preserved by rational translation, it suffices to prove it when $C=\cone(r_1,\dots,r_k)$.
Then, $\dom(C)=\cone(r_1,\dots,r_k,\chi^1,\dots,\chi^n)$, hence its polar cone $\dom(C)^*$ is $C^*\cap\{x\leq \mathbf{0}\}$. 
If $C$ is box-TDI, then so is $C^*$ by Lemma~\ref{lem:polyBTDIminconesBTDI}, hence so is $C^*\cap\{x\leq \mathbf{0}\}=\dom(C)^*$, and, by Lemma~\ref{lem:polyBTDIminconesBTDI} again, so is $\dom(C)$.
\end{proof}

\begin{csq}[{\cite[{Remark 2.21}]{sc}}]\label{affPTU}
If $P$ is a box-TDI polyhedron, then $\aff(P)=\{x:Cx=d\}$ for some totally unimodular matrix~$C$.
\end{csq}
\begin{proof}
If $P$ is a box-TDI polyhedron, then by point~\ref{polyiii} of Theorem~\ref{cor:main}, since $P$ is a face of $P$, its affine hull can be described using a totally unimodular matrix.
\end{proof}

%\begin{csq}[{\cite[{Theorem 2.16}]{edgi}}]\label{pmdescrfulldim}
%If $P$ is a full-dimensional box-TDI polyhedron, then $P =\{x\in\mathbb{R}^n:Ax \geq b\}$ for some $\zpmset$-matrix $A$ and some vector $b$.
%\end{csq}
%\begin{proof}
%If $ax\leq b$ is facet-defining for $P$, then, by Theorem~\ref{main}, all the nonzero coefficients of $a$ have the same absolute value $d$, hence $\frac{1}{d} a x\leq \frac{b}{d}$ has $\zpmset$ coefficients and define the same facet as $ax\leq b$.
%\end{proof}

\begin{csq}[{\cite[{Remark 2.22}]{sc}}]
Each edge and each extremal ray of a pointed box-TDI polyhedron is in the direction of a $\zpmset$-vector.
\end{csq}
\begin{proof}
This is point~\ref{polylinF} of Corollary~\ref{cor:pol} applied to the faces of dimension one of the polyhedron.
\end{proof}

By polarity, the above proof shows that every full-dimensional box-TDI polyhedron can be described using a $\zpmset$-matrix. Edmonds and Giles prove in~\cite{edgi} that it is still true without the full-dimensional hypothesis.

\begin{csq}[{\cite[{Theorem 2.16}]{edgi}}]\label{pmdescr}
If $P$ is a box-TDI polyhedron, then $P =\{x:Ax \leq b\}$ for some $\zpmset$-matrix $A$ and some vector $b$.
\end{csq}
\begin{proof}
Let $P$ be a box-TDI polyhedron. By Consequence~\ref{affPTU}, we have $\aff(P)=\{x:Cx=d\}$ for some full row rank totally unimodular matrix $C$. By point~\ref{polyiii} of Theorem~\ref{cor:main}, for each facet $F$ of~$P$, there exists a totally unimodular matrix $D_F$ such that $\aff(F)=\{x:D_Fx=d_F\}$. Then, one of the rows $a_F x= b_F$ of $D_Fx=d_F$ does not contain $\aff(P)$. Possibly multiplying by $-1$, we may assume that $a_F x\leq b_F$ is valid for $P$ because $F$ is a facet of $P$. Then, the matrix $A$ whose rows are those of $C$ and every $a_F$ yields a description of $P$ as desired.
\end{proof}

\subsubsection{Cook's Characterization~\cite{co}, {\cite[Theorem~22.9]{sc}}}\label{ss:boxhilbert}

In order to get a geometric characterization of box-TDI polyhedra, Cook~\cite{co} introduced the so-called box property. 
Schrijver~\cite[Theorem~22.9]{sc} states Cook's characterization with the following equivalent form of the box property: a cone $C$ of $\mathbb{R}^n$ has the {\em box property} if for all $c\in C$ there exists $\widetilde c\in C\cap \mathbb{Z}^n$ such that $\lfloor c \rfloor \leq \widetilde c \leq \lceil c \rceil$.
To hightlight the connections with our results, we reformulate Schrijver's version as follows.

\begin{itemize}
	\item A polyhedron is box-TDI if and only if the normal cones of its faces all have the box property (Cook~\cite[Theorem~22.9]{sc}).
\end{itemize}
The parallel with our work is clear with the following reformulation of one of our characterizations.

\begin{itemize}
	\item A polyhedron $P$ is box-TDI if and only if every minimal tangent cone of $P$ is box-integer, up to translation (Observation~\ref{biconesprop} and Theorems~\ref{pbiEM} and~\ref{main}).
\end{itemize}
The first difference between these two results is that the first one involves the normal cones, whereas the second one involves the tangent cones. 
Recall that the tangent cones are the polars of the normal cones, up to translation. 
This polarity connection between the two statements is not surprising in light of the polarity result of Lemma~\ref{boxtdicone}.
The second difference is that the first result involves the box property, whereas the second involves the notion of box-integrality.
It is easy to see that box-integer cones have the box property. 
The converse does not hold. 
In fact, the lemma below shows that the box property is a local property when the box-integrality is a global one.
The third difference is a consequence of this local/global aspect: the first result involves {\em all} the normal cones, whereas the second involves only the {\em minimal} tangent cones.

To sum up, the first result is a polar local characterization of box-TDI polyhedra, and the second is a primal global characterization.

\begin{proposition}
A cone $C=\{x:Ax\leq \mathbf 0\}$ is box-integer if and only if all its faces have the box property.
\end{proposition}
The following lemma proves the proposition, since a cone $C$ is box-integer if and only if $\aff(F)$ is box-integer for all faces $F$ of $C$.

\begin{lemma}\label{obs:boxintboxprop} 
Let $C=\{x:Ax\leq \mathbf 0\}$ and let $F=\{x:A_F x\leq \mathbf 0\}$ be a face of~$C$.
\begin{itemize}
\item If $C$ is box-integer, then $F$ has the box property.
\item If $F$ has the box property, then $\aff(F)$ is box-integer.
\end{itemize}
\end{lemma}
\begin{proof}
Suppose that  $C$ is box-integer and  let $c\in F$. Since $c$ belongs to $P = F \cap \{\lfloor c \rfloor \leq x \leq \lceil c \rceil\}$, the latter is nonempty. Since $C$ is box-integer, so is $F$, hence $P$ has only integer vertices, and any of them forms a suitable $\widetilde c$ which shows that $F$ has the box property.

Suppose now that $F$ has the box property. Let $p\in\mathbb Z^I$ be such that $\aff(F)\cap \{x_i=p_i,i\in I\}$ is a singleton $c$ in $\aff(F)$. There exists $t\in\mathbb Z^n$ be such that $c'=c+t\in F$. By the box property of $F$, there exists $\widetilde c\in F\cap \mathbb{Z}^n$ such that $t+\lfloor c \rfloor=\lfloor c' \rfloor \leq \widetilde c \leq \lceil c' \rceil=\lceil c\rceil+t$. Now, $\widetilde c-t$ belongs to $\aff(F)\cap \{x_i=p_i,i\in I\}$, hence $c=\widetilde c-t$ is integer. By Lemma~\ref{boxint}, $\aff(F)$ is box-integer.
\end{proof}

In a way, the above lemma shows that the box property of a cone is sandwiched between the box-integrality of the cone and that of its underlying affine space---an even more local property. This, up to polarity again, further compares Cook's characterization and ours, as the latter property appears in Lemma~\ref{lm:pbiaff}.

The following picture illustrates some differences between the three properties.
\begin{figure}[h!]
\centering
\begin{tikzpicture}
\coordinate (Origin) at (0,0);
\coordinate (XAxisMin) at (-.5,0);
\coordinate (XAxisMax) at (4.5,0);
\coordinate (YAxisMin) at (0,-.5);
\coordinate (YAxisMax) at (0,2.5);
\draw [thin, gray,-latex] (XAxisMin) -- (XAxisMax);% Draw x axis
\draw [thin, gray,-latex] (YAxisMin) -- (YAxisMax);% Draw y axis
\clip (-.5,-.5) rectangle (4.5cm,2.5cm); % Clips the picture...
\draw[style=help lines,dashed] (-1,-1) grid[step=1cm] (5,5);
\draw[fill=gray, fill opacity=0.3, draw=black] (Origin) -- (8,4) -- (8,0) -- cycle;
\draw [pattern=north west lines, fill opacity=0.3, draw=black] (Origin) -- (6,2) -- (8,4) -- cycle;
\draw [ultra thick,-latex] (Origin) -- (2,1) node [above left] {};
\draw [ultra thick,-latex] (Origin) -- (3,1) node [above left] {};
\draw [ultra thick,-latex] (Origin) -- (1,0) node [below right] {};
% Draws a grid in the new coordinates.
%\filldraw[fill=gray, fill opacity=0.3, draw=black] (0,0) rectangle (2,2);
% Puts the shaded rectangle
\foreach \x in {-1,0,...,4}{% Two indices running over each
\foreach \y in {-1,0,...,2}{% node on the grid we have drawn
\node[draw,circle,inner sep=2pt,fill] at (\x,\y) {};
% Places a dot at those points
}
}
\end{tikzpicture}
\caption{
The cone $C=\cone\{(2,1),(1,0)\}$ has the box property but is not box-integer.
The cone $C'=\cone\{(2,1),(3,1)\}$ does not have the box property, yet $\aff(C')=\mathbb R^2$ is box-integer.
The cone $C''=\cone\{(2,1)\}$ does not have the box property, yet its polar does.}
\label{figboxH}
\end{figure}
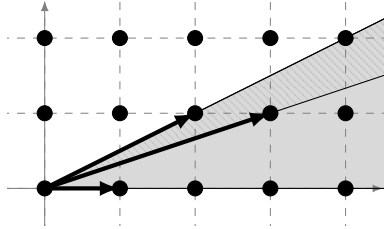

The notion of box-integrality of cones and affine spaces sheds a better light on box-TDI polyhedra by providing insights of how their local, global, and polar properties are connected.
Both are preserved by polarity, the global notion yields a global geometric characterization of box-TDI polyhedra, and the most local one allows to derive matricial counterparts.

\subsubsection{Schrijver's Sufficient Condition~\cite[Theorem~5.35]{scbig}}\label{ss:schri}
In this section, we compare our results on box-TDI polyhedra with known results on box-TDI systems.
It appears that our results in some sense allow to split the ``box-'' from the ``-TDI'': to prove that a given system is box-TDI, prove that it is TDI on the one hand, and prove that the polyhedron is principally box-integer on the other hand.

\medskip
 
As noticed by Schrijver~\cite[Page~318]{sc}, Hoffman and Kruskal's result~\cite{hokr} implies that a matrix $A$ is totally unimodular if and only if the system $Ax\leq b$ is box-TDI for each vector $b$. Then, by Theorem~\ref{TEMPBI} and Theorem~\ref{main}, the parallel with totally equimodular matrices can be thought as relaxing the box-TDIness of those systems to that of the associated polyhedra.

%Before discussing Schrijver's sufficient condition, we make a digression about totally equimodular matrices. As noticed by Schrijver~\cite[Page~318]{sc}, Hoffman and Kruskal's result~\cite{hokr} implies the following characterization: a matrix $A$ is totally unimodular if and only if the system $Ax\leq b$ is box-TDI for each vector $b$. Then, by Theorem~\ref{TEMPBI} and Theorem~\ref{main}, the parallel with totally equimodular matrices can be thought as relaxing the box-TDIness of those systems to that of the associated polyhedra.

\begin{corollary}\label{TEBTDI}
A matrix $A$ of $\mathbb{Q}^{m\times n}$ is totally equimodular if and only if the polyhedron $\{x: Ax\leq b\}$ is box-TDI for all $b\in\mathbb{Z}^m$.
\end{corollary}

Totally unimodular matrices being totally equimodular, the following well-known result is a special case of the above corollary. 

\begin{csq}%[{}]
A polyhedron whose constraint matrix is totally unimodular is box-TDI.
\end{csq}

We mention that there exist box-TDI systems which are not defined by a totally unimodular matrix. 
By Corollary~\ref{TEBTDI} and Theorem~\ref{tdiboxtdisystem}, any TDI system defined with a totally equimodular matrix is box-TDI. 
Therefore, to find a box-TDI system for a polyhedron described by a totally equimodular matrix, there only remains to find a TDI system describing this polyhedron.

\medskip

Another interesting parallel can be observed with Schrijver's Sufficient Condition. Schrijver proves in~\cite[Theorem~5.35]{scbig} that the following weakening of $A$ being totally unimodular already suffices to obtain the box-TDIness of the system $Ax\leq b$. 

\begin{theorem}[{\cite[Theorem~5.35]{scbig}}]\label{scTU}
Let $Ax\leq b$ be a system of linear inequalities, with $A$ an $m\times n$ matrix. Suppose that \textup{($\star$)} for each $c \in\mathbb{R}^n$, $\max\{c^\top x : Ax\leq b\}$ has (if finite) an optimum dual solution $y \in\mathbb{R}^m_+$ such that the rows of $A$ corresponding to positive components of $y$ form a totally unimodular submatrix of $A$. Then $Ax\leq b$ is box-TDI.
%If all the submatrices of $A$ which are face-defining for $P=\{x:Ax\leq b\}$ are TU, then $Ax\leq b$ is a box-TDI system.
\end{theorem}

Note that the property~\textup{($\star$)} is equivalent to the condition that for every face $F$ of $\{x:Ax\leq b\}$, the system $Ax\leq b$ contains a totally unimodular face-defining matrix for $F$. 
Theorem~\ref{cor:main} contains a polyhedral version: a polyhedron is box-TDI if and only if each of its faces has a totally unimodular face-defining matrix. 
This latter condition is weaker than~\textup{($\star$)}, hence does not ensure the box-TDIness of the system. Nevertheless, when satisfied, all that remains to do is to find a TDI system describing the same polyhedron.

In light of our characterizations, one could wonder whether Theorem~\ref{scTU} can be turned into an equivalence, that is: can every box-TDI polyhedron be described by a box-TDI system satisfying~\textup{($\star$)}? 
Unfortunately, the answer to this question is negative.
Indeed, systems satisfying~\textup{($\star$)} can be assumed $\zpmset$, and there exist box-TDI polyhedra for which no TDI description is $\zpmset$, see~\cite[Page~325]{sc}.

\section{Illustrations}\label{s:illustrations}

In this section, we provide illustrations of our results. The first one is a new perspective on the equivalence between two results about binary clutters. Secondly, we refute a conjecture of Ding, Zang, and Zhao~\cite{dizazh} about box-perfect graphs. Thirdly, we discuss connections with an abstract class of polyhedra introduced in~\cite{gire}. Finally, we characterize the box-TDIness of the cone of conservative functions of a graph.

\subsection{Box-Mengerian Clutters}

We briefly introduce the definitions we need  about clutters. 
A collection $\mathcal{C}$ of subsets of a set $E$ is a {\em clutter} if none of its sets contains another one. 
We denote by $A_\mathcal{C}$ the $\mathcal{C}\times E$ incidence matrix of $\mathcal{C}$ and by $P_\mathcal{C}=\{x\in\mathbb{R}^{E}:A_\mathcal{C} x\geq 1, x\geq\mathbf{0}\}$ the associated covering polyhedron. 
A clutter $\mathcal{C}$ is {\em binary} if the symetric difference of any three elements of $\mathcal{C}$ contains an element of $\mathcal{C}$.
A clutter $\mathcal{C}$ is {\em box-$\frac 1d$-integral} if for all $\ell,u\in\frac1d\mathbb{Z}^{E}$, each vertex of $P_\mathcal{C}\cap\{\ell\leq x\leq u\}$ belongs to $\frac1d \mathbb{Z}^{E}$. 
%Note that a clutter $\mathcal{C}$ is box-$\frac 1d$-integral if and only if $dP_\mathcal{C}$ is box-integer.
% Je suis pas certain que ce soit intéressant de laisser ça là ?
A matrix $A\in\{0,1\}^{m\times n}$ is called {\em (box-)Mengerian} if the system $Ax\geq 1, x\geq\mathbf{0}$ is (box-)TDI. A clutter $\mathcal{C}$ is {\em (box-)Mengerian} if $A_\mathcal{C}$ is (box-)Mengerian.
{\em Deleting} an element $e\in E$ means replacing $\mathcal{C}$ by $\mathcal{C}\setminus e=\{X\in\mathcal{C}:e\notin X\}$ and {\em contracting} an element $e\in E$ means replacing $\mathcal{C}$ by $\mathcal{C}/e$ which is composed of the inclusionwise
 minimal members of $\{X\setminus\{e\}:X\in\mathcal{C}\}$. The {\em minors} of a clutter are the clutters obtained by repeatedly deleting and contracting elements of $E$.
The clutter $Q_6$ is defined on the set $E_4$ of the edges of the complete graph $K_4$, and its elements are the triangles
 of $K_4$---see Figure~\ref{figQ6}.
 The clutter $Q_7$ is defined on $E_4\cup e$ where $e\notin E_4$, and its elements are $X\cup \{e\}$ for each triangle or perfect matching $X$ of $K_4$.

\medskip

In 1995, Gerards and Laurent~\cite{gela} characterized the binary clutters that are box-$\frac 1d$-integral for all $d\in\mathbb{Z}_{\scriptscriptstyle>0}$ by forbidding minors.

\begin{theorem}[{\cite[Theorem~1.2]{gela}}]\label{1dbox}
A binary clutter is box-$\frac 1d$-integral for all $d\in\mathbb{Z}_{\scriptscriptstyle>0}$ if and only if neither $Q_6$ nor $Q_7$ is its minor.
\end{theorem}
%
%Seymour~\cite{} characterizes Mengerian binary clutters.
%
%\begin{theorem}\label{seymourM}
%A binary clutter is Mengerian if and only if does not contain $Q_6$ as a minor.
%\end{theorem}

In 2008, Chen, Ding, and Zang~\cite{chdiza} characterized box-Mengerian binary clutters by forbidding minors. In~\cite{chchza}, Chen, Chen, and Zang provide a simpler proof of this characterization, based on the so called ESP property. We mention that none of the proofs of Theorem~\ref{boxM} rely on Theorem~\ref{1dbox}.

\begin{theorem}[{\cite[Corollary~1.2]{chdiza}}]\label{boxM}
A binary clutter is box-Mengerian if and only if neither $Q_6$ nor $Q_7$ is its minor.
\end{theorem}

The combination of Theorems~\ref{1dbox} and~\ref{boxM} implies that a binary clutter is box-Mengerian if and only if it is box-$\frac 1d$-integral for all $d\in\mathbb{Z}_{\scriptscriptstyle>0}$. We show in the following how this equivalence is actually a special case of Theorem~\ref{main}.

By definition, a clutter $\mathcal{C}$ is box-$\frac 1d$-integral if and only if $dP_\mathcal{C}$ is box-integer, which implies the following reformulation of the class of polyhedra characterized in Theorem~\ref{1dbox}.

\begin{center}
A clutter $\mathcal C$ is box-$\frac1d$-integral for all $d\in\mathbb{Z}_{\scriptscriptstyle>0}$ if and only if $P_{\mathcal C}$ is fully box-integer. 
\end{center}

Recall that a system is box-TDI if and only if it is TDI and defines a box-TDI polyhedron. Then, by Theorem~\ref{main}, a clutter is box-Mengerian if and only if it is Mengerian and $P_\mathcal C$ is principally box-integer. Since $\mathcal{C}$ being Mengerian implies the integrality of $P_\mathcal{C}$,
 we get the following reformulation for the systems involved in Theorem~\ref{boxM}.

\begin{center}
A clutter $\mathcal{C}$ is box-Mengerian if and only if it is Mengerian and $P_\mathcal{C}$ is fully box-integer.
\end{center}

Therefore, to prove the announced equivalence it is enough to show the following statement.

\begin{center}
If $\mathcal C$ is binary and $P_\mathcal C$ is fully box-integer, then $\mathcal C$ is Mengerian.
\end{center}

We apply Seymour's characterization~\cite{se2}: a binary clutter is Mengerian if and only if it has no $Q_6$ minor. 
The property of $P_\mathcal{C}$ being fully box-integer is closed under taking minors since $P_{\mathcal{C}/e}$ and $P_{\mathcal{C}\setminus e}$ are respectively obtained from $P_{\mathcal{C}}\cap\{x_e=0\}$ and $P_{\mathcal{C}}\cap\{x_e=1\}$ by deleting $e$'s coordinate. Furthermore, $P_{Q_6}$ is not fully box-integer by point~\ref{pbiEM2} of Theorem~\ref{pbiEM}. Indeed, the first three rows of the matrix $A_{Q_6}$ of Figure~\ref{figQ6} form a nonequimodular matrix $M$, as the determinant of the three first columns equals $2$ and that of the three last columns equals $1$. 
Moreover, $M$ is face-defining for $P_{Q_6}$, by Observation~\ref{face-def_C} and because $\chi^1+\chi^6$, $\chi^2+\chi^5$, $\chi^3+\chi^4$, and $\chi^4+\chi^5+\chi^6$ are affinely independent, belong to $P_{Q_6}$, and satisfy $Mx=1$. 
Therefore, if $\mathcal C$ is binary and $P_\mathcal{C}$ is fully box-integer, then $\mathcal{C}$ has no $Q_6$ minor.

\begin{figure}[h!]
\begin{center}
$
A_{Q_6}=\left[\begin{array}{cccccc}
	1 & 1 & 0 & 1 & 0 & 0\\
	1 & 0 & 1 & 0 & 1 & 0\\
	0 & 1 & 1 & 0 & 0 & 1\\
	0 & 0 & 0 & 1 & 1 & 1
\end{array}\right]
$
\caption{The matrix representation of the clutter $Q_6$.}\label{figQ6}
\end{center}
\end{figure}

\subsection{On Box-Perfect Graphs}

In this section, we provide a construction which preserves non box-perfection, and use it to refute a conjecture of Ding, Zang, and Zhao~\cite{dizazh}.

In a graph, a {\em clique} is a set of pairwise adjacent vertices, and a {\em stable set} is the complement of a clique. The {\em stable set polytope} of a graph is the convex hull of the incidence vectors of its stable sets. {\em Perfect graphs} are known to be those whose stable set polytope is described by the system composed of the clique inequalities and the nonnegativity constraints:
$$
%\text{(\/Stable)}
%\left\{
\begin{array}{ll}
x(C) \leq 1 & \quad \textrm{for all cliques $C$,}\\
x \geq \mathbf{0}.
\end{array}
%\right.
$$
A {\em box-perfect graph} is a graph for which this system is box-TDI. 
Since this system is known to be TDI if and only if the graph is perfect~\cite{lo}, a graph is box-perfect if and only if it is perfect and its stable set polytope is box-TDI.
The characterization of box-perfect graphs is a long standing open question raised by Cameron and Edmonds in 1982~\cite{caed}. Recent progress has been made on this topic by Ding, Zang, and Zhao~\cite{dizazh}. They exhibit several new subclasses of perfect graphs, and in particular prove the conjecture of Cameron and Edmonds~\cite{caed} that parity graphs are box-perfect. They also propose a conjecture for the characterization of box-perfect graphs. 

\medskip

To state their conjecture, they introduce the class of graphs $\mathcal{R}$, built as follows. 
Let $G=(U,V,E)$ be a bipartite graph whose biadjacency matrix is minimally non-TU. Add a set of edges $F$ between vertices of $V$ such that the neighborhood $N_{G'}(u)$ of $u$ in $G'=(U\cup V,E\cup F)$ is a clique for all $u\in U$. If there exists $u\in U$ such that $N_{G'}(u)=V$, then $G'\setminus\{u\}$ is in $\mathcal{R}$, otherwise $G'$ is in~$\mathcal{R}$.

\begin{conjecture}[Ding, Zang, and Zhao~\cite{dizazh}]\label{conjboxperfect}
A perfect graph is box-perfect if and only if it contains no graph from $\mathcal{R}$ as an induced subgraph.                               
\end{conjecture}

We introduce the operation of {\em unfolding} a vertex $v\in V$ in $G=(V,E)$.
Take a vertex $v\in V$ and two sets of vertices $X$ and $Y$ such that $X\cup Y = N_G(v)$ and no edge connects $X\setminus Y$ and $Y\setminus X$.
Delete $v$ and add two new vertices $x$ and $y$ such that the neighborhoods of $x$ and $y$ are respectively $X$ and $Y$.
 Finally, add another vertex $z$ adjacent only to $x$ and $y$. 

We mention that unfolding a vertex might not preserved perfection.
 Nevertheless, if the starting graph is perfect but not box-perfect, then the graph obtained by unfolding is not box-perfect.

\begin{lemma}\label{replacement1}
Unfolding any vertex in a perfect but not box-perfect graph yields a non box-perfect graph.
\end{lemma}

\begin{proof}
We show that if the stable set polytope of a graph has a nonequimodular face-defining matrix, then so does any graph obtained by unfolding. By Theorem~\ref{cor:main}, this proves the Lemma.

Let $G=(V,E)$ be a graph which is perfect but not box-perfect, let $v$ be a vertex of $G$, let $H$ be obtained from $G$ by unfolding $v$, and $x,y,z$ be the new vertices. Let $n=|V|$. Since $G$ is not box-prefect, its stable set polytope has a nonequimodular face-defining matrix $M\in\mathbb{Q}^{k\times n}$ for a face~$F$. Since $G$ is perfect, we may assume that the rows of $M$ are the incidence vectors of a set $\mathcal{K}$ of cliques of $G$. Indeed, it can be checked that removing the rows corresponding to nonnegativity constraints yields a smaller nonequimodular face-defining matrix. By Observation~\ref{face-def_C}, there exists a family $\mathcal{S}$ of affinely independent stable sets of $F$ with $|\mathcal{S}|=n-\dim(F)+1$. Build a family $\mathcal{T}$ of stable sets of $H$ from $\mathcal{S}$ as follows: if $S\in\mathcal{S}$ contains $v$, then $S\setminus\{v\}\cup\{x,y\}\in\mathcal{T}$, otherwise $S\cup\{z\}\in\mathcal{T}$. All these sets are stable sets and are affinely independent. 
Build a family $\mathcal{L}$ of $k+2$ cliques of $H$ as follows. For each $K\in\mathcal{K}$, 
\begin{itemize}
	\item If $v\notin K$, then $K\in \mathcal{L}$.
	\item If $v\in K$, the fact that $X\cup Y=N_G(v)$ and no edge connects $X\setminus Y$ and $Y\setminus X$ ensures that at least one of $K\setminus\{v\}\cup \{x\}$ and $K\setminus\{v\}\cup \{y\}$ is a clique of $H$. If both are cliques, then add one of them to $\mathcal{L}$, otherwise add the clique.
	\item Add $\{x,z\}$ and $\{y,z\}$ to $\mathcal{L}$.
\end{itemize}
Let $N$ denote the $(k+2)\times (n+2)$ matrix whose rows are the incidence vectors of the cliques of~$\mathcal{L}$. 
The matrix $N$ has full row rank and each stable set $T$ of $\mathcal{T}$ satisfies $|T\cap L|=1$ for all $L\in\mathcal{L}$, hence $N$ is face-defining for the stable set polytope of $H$ by Observation~\ref{face-def_C}. 
There only remains to show that $N$ is not equimodular. 
To prove this, we show that each $k\times k$ submatrix of $M$ gives rise to a $(k+2)\times (k+2)$ submatrix of $N$ having the same determinant.
Since $M$ is not equimodular, neither is $N$.

Let $A$ be a $k\times k$ submatrix of $M$. If $A$ does not contains $v$'s column $M^v$, then add two rows of zeros and then the two columns $N^y$ and $N^z$. Note that the determinant has not changed: first develop with respect to $\{x,z\}$'s row, and then with respect to $\{y,z\}$'s row, to obtain the starting matrix. If $A$ contains $v$'s column $M^v$, then delete it, add two rows of zeros and finally add the three columns $N^x$, $N^y$, and $N^z$. Let $A'$ denote this new matrix. We obtain $\det(A')=\det(A)$ as follows: first replace the column $A^x$ by $A^x +A^y-A^z$, then develop with respect to $\{x,z\}$'s row, and finally with respect to $\{y,z\}$'s row. The resulting matrix is precisely~$A$.
\end{proof}

Unfolding a vertex in $S_3$ as shown in Figure~\ref{fig:nonboxperfect} yields a graph which is perfect but not box-perfect, and contains no induced subgraphs from $\mathcal{R}$. This disproves Conjecture~\ref{conjboxperfect}---see Proposition~\ref{boxpfctrex}.

%\begin{lemma}\label{replacement2}
%Replacing any vertex in a graph from $\mathcal{R}$ yields a graph containing no graph from $\mathcal{R}$.
%\end{lemma}
%\begin{proof}
%Let $L=(W,D)$ be a graph from $\mathcal{R}$ and let $H$ be obtained from $L$ by replacing $v$ for some $v\in W$. Let $x,y,z$ be the resulting vertices. Suppose that $H$ belongs to $\mathcal{R}$ and is built from a bipartite graph $G=(U,V,E)$ and a set of edge $F$. Note that, since $B_G$ is minimally non TU, the graph $G$ is connected. Therefore, since $N_H(z)=\{x,y\}$ is not a clique, we have $z\in V$. Therefore $x$ and $y$ belong to $U$. Then, since $G$ is connected and $|V|\geq |U|$, either $x$ or $y$ has another neighbor $w$ in $V$. But then $w$ should be adjacent to $z$, a contradiction.
%
%Prove that $H$ contains no graph from $\mathcal{R}$.
%\end{proof}

\begin{center}
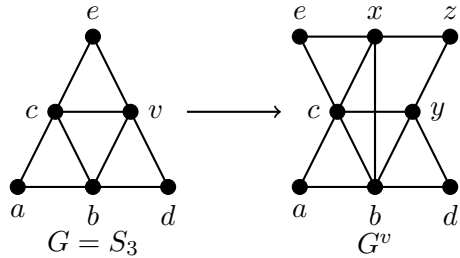
\begin{figure}[h!]
\centering
  \begin{tikzpicture}[scale=.5]
     \coordinate (Origin)   at (0,0);
     \node at (0,-1.5) {$G=S_3$};

     \node[draw,circle,inner sep=2pt,fill] at (-2,0) {};
     \node[below=.1cm] at (-2,0) {$a$};
     \node[draw,circle,inner sep=2pt,fill] at (0,0) {};
     \node[below=.1cm] at (0,0) {$b$};		
     \node[draw,circle,inner sep=2pt,fill] at (2,0) {};
     \node[below=.1cm] at (2,0) {$d$};
		
     \node[draw,circle,inner sep=2pt,fill] at (-1,2) {};
     \node[left=.1cm] at (-1,2) {$c$};
     \node[draw,circle,inner sep=2pt,fill] at (1,2) {};
     \node[right=.1cm] at (1,2) {$v$};
		
     \node[draw,circle,inner sep=2pt,fill] at (0,4) {};
     \node[above=.1cm] at (0,4) {$e$};
		
		 \draw[thick] (-2,0) -- (0,0) -- (2,0) -- (1,2) -- (0,0) -- (-1,2) -- (1,2) -- (0,4) -- (-1,2) {};
		 \draw[thick] (-2,0) -- (-1,2) {};
		
		 \draw[thick,->] (2.5,2) -- (5,2) {};
		
		\begin{scope}[xshift=+7.5cm]
     \node at (0,-1.5) {$G^v$};

     \node[draw,circle,inner sep=2pt,fill] at (-2,0) {};
     \node[below=.1cm] at (-2,0) {$a$};
     \node[draw,circle,inner sep=2pt,fill] at (0,0) {};
     \node[below=.1cm] at (0,0) {$b$};		
     \node[draw,circle,inner sep=2pt,fill] at (2,0) {};
     \node[below=.1cm] at (2,0) {$d$};
		
     \node[draw,circle,inner sep=2pt,fill] at (-1,2) {};
     \node[left=.1cm] at (-1,2) {$c$};
     \node[draw,circle,inner sep=2pt,fill] at (1,2) {};
     \node[right=.1cm] at (1,2) {$y$};
		
     \node[draw,circle,inner sep=2pt,fill] at (2,4) {};
     \node[above=.1cm] at (2,4) {$z$};
		
     \node[draw,circle,inner sep=2pt,fill] at (0,4) {};
     \node[above=.1cm] at (0,4) {$x$};
		
     \node[draw,circle,inner sep=2pt,fill] at (-2,4) {};
     \node[above=.1cm] at (-2,4) {$e$};
		
		 \draw[thick] (-2,0) -- (0,0) -- (2,0) -- (1,2) -- (0,0) -- (-1,2) -- (1,2) {};
		 \draw[thick] (-2,0) -- (-1,2) {};
		 \draw[thick] (0,4) -- (0,0) {};
		 \draw[thick] (1,2) -- (2,4) -- (0,4) -- (-2,4) -- (-1,2) -- (0,4) {};
		\end{scope}
	\end{tikzpicture}
	\caption{A non box-perfect graph obtained by unfolding the vertex $v$ in $S_3$, with $X=\{b,c,e\}$ and $Y=\{b,c,d\}$.}\label{fig:nonboxperfect}
\end{figure}
\end{center}

It is well known that the graph $S_3$ in Figure~\ref{fig:nonboxperfect} is not box-perfect~\cite{caca}. It can also be seen because the nonequimodular matrix $M$ below is face-defining for the stable set polytope of $S_3$. 

$$
M=\left[\begin{array}{cccccc}
	1 & 1 & 0 & 1 & 0 & 0\\
	1 & 0 & 1 & 0 & 1 & 0\\
	0 & 1 & 1 & 0 & 0 & 1
\end{array}\right]
$$
Indeed, up to reordering the vertices, the rows of $M$ correspond to the three external triangles, and the four affinely independent stable sets $\{a,v\},\{b,e\},\{c,d\},\{a,d,e\}$ belong to the corresponding face. By Observation~\ref{face-def_C} and Theorem~\ref{cor:main}, $S_3$ is not box-perfect.

\begin{proposition}\label{boxpfctrex}
The graph $G$ of Figure~\ref{fig:nonboxperfect} is perfect but not box-perfect and none of its induced subgraphs belongs to $\mathcal{R}$.
\end{proposition}
\begin{proof}
Note that the graphs $G$ and $G^v$ are perfect. By Lemma~\ref{replacement1}, $G^v$ is not box-perfect. The graph $G^v\setminus \{z\}$ is box-perfect, as one can check that the constraint matrix of its stable set polytope is totally unimodular. Hence, if $G^v$ contains an induced subgraph $H\in\mathcal{R}$, then $z\in V(H)$. As no graph in $\mathcal R$ has a vertex of degree one, this contradicts the claim below.  
\begin{center}
If $H\in\mathcal{R}$ has a vertex $z$ with only two neighbors $x$ and $y$, then $xy$ is an edge of $H$.
\end{center}
Suppose that $H$ is built from a bipartite graph $K=(U,V,E)$ and a set of edges $F$. Note that, since the biadjacency matrix of $K$ is minimally non-TU, the graph $K$ is connected. Suppose that $xy$ is not an edge of $H$. Then, $N_H(z)=\{x,y\}$ is not a clique. Thus $z\in V$, and $x$ and $y$ belong to $U$. Then, since $K$ is connected and $|V|\geq |U|$, either $x$ or $y$ has another neighbor $w$ in $V$. But then $w$ should be adjacent to $z$, a contradiction.
\end{proof}

Note that chosing $X=\{c,e\}$ and $Y=\{b,c,d\}$ when unfolding $v$ in Figure~\ref{fig:nonboxperfect} yields another perfect but not box-perfect graph with no graph from $\mathcal{R}$ as an induced subgraph.

\subsection{Integer Decomposition Property}

In this section, we discuss possible connections between full box-integrality and the integer decomposition property. 
This property arises in various fields such as integer programming, algebraic geometry, combinatorial commutative algebra. 
 %A recent important result on this topic is that polytopes with long edges have the integer decomposition property. 
% Beside, 
Several classes of polyhedra are known to have the integer decomposition property, as for instance: projections of polyhedra defined by totally unimodular 
 matrices~\cite{se}, polyhedra defined by nearly totally unimodular 
 matrices~\cite{gi}, certain polyhedra defined by $k$-balanced matrices~\cite{za}, the stable set polytope of claw-free $t$-perfect graphs and $h$-perfect line-graphs~\cite{be}.

\medskip

A polyhedron~$P$ has the {\em integer decomposition property}, if for any natural number~$k$ and any integer vector $x\in kP$, there exist $k$ integer vectors $x_1,\dots,x_k\in P$ with $x_1+\dots+x_k=x$. 
A stronger property is when the polyhedron~$P$ has the {\em Integer Carath\'eodory Property}, that is, if for every positive integer~$k$ and every integer vector $x\in kP$, there exist $n_1,\dots,n_t \in \mathbb{Z}{\scriptscriptstyle \geq 0}$ and affinely independent $x_1,\dots,x_t\in P\cap\mathbb{Z}^n$ such that $n_1+\dots+n_t =k$ and $x = \sum_i n_ix_i$.

\medskip 

In~\cite{gire}, Gijswijt and Regts introduce a class~$\mathcal{P}$ of polyhedra and show that they have the Integer Carath\'eodory Property. 
%This class for instance contains (poly)matroid base polytopes.  
They define~$\mathcal{P}$ to be the set of polyhedra~$P$ such that for any $k\in\mathbb{Z}_{\scriptscriptstyle \geq 0}$, $r \in\{0,\dots,k\}$, and $w\in\mathbb{Z}^n$ the intersection $rP \cap \left(w - (k -r)P\right)$ is box-integer.
They also show~\cite[Proposition~4]{gire} that every $P\in\mathcal{P}$ is box-integer. 
Given the definition of~$\mathcal{P}$, note that if a polyhedron is in~$\mathcal{P}$, then so are all its dilations. Therefore, every $P$ in $\mathcal{P}$ is fully box-integer.
By Theorem~\ref{main}, this has the following consequence.

\begin{corollary}\label{pinp}
Every $P\in\mathcal{P}$ is box-TDI.
\end{corollary}

The converse of Corollary~\ref{pinp} does not hold. 
We show below that polyhedra in $\mathcal{P}$ satisfy the stronger property that not only the affine hulls of their faces are principally box-integer, but also the intersection of the affine hulls of any two faces.
In terms of matrices, this is phrased as follows.

\begin{proposition}
If $P\in\mathcal{P}$, then $\aff(F)\cap\aff(G)$ has an equimodular face-defining matrix for all faces $F$ and $G$ of $P$.
\end{proposition}
\begin{proof}
Let $F$ and $G$ be faces of $P$, and let $x_F$ and $x_G$ be rational points in their respective relative interior. 
There exists $k\in\mathbb{Z}_{\scriptscriptstyle > 0}$ such that both $kx_F$ and $kx_G$ are integer. 
Let $w=k(x_F+x_G)$, and $Q = kP\cap (w - kP)=k(P\cap (x_F+x_G-P))$. Since $P\in\mathcal{P}$, note that $r Q$ is box-integer for all $r\in\mathbb{Z}_{\scriptscriptstyle > 0}$, that is, $Q$ is fully box-integer.
By the choice of $x_F$ and $x_G$, the minimal face $H$ of $Q$ containing $kx_F$ satisfies $\aff(H)=k\left(\aff(F)\cap -(x_F+x_G+\aff(G))\right)$. 
Thus, the latter is a translation of $\aff(F)\cap -\aff(G)$. 
Since $Q$ is fully box-integer, $\aff(H)$ has an equimodular face-defining matrix by Theorem~\ref{pbiEM}, hence so has $\aff(F)\cap -\aff(G)$ by translation. 
Since $\aff(F)\cap \aff(G)$ can be described using the matrix of constraints of $\aff(F)\cap \aff(G)$ and multiplying by $-1$ the right-hand sides corresponding to $\aff(G)$, we get an equimodular face-defining matrix for  $\aff(F)\cap \aff(G)$.
\end{proof}

Fully box-integer polyhedra do not inherit the Integer Carath\'eodory Property. Actually, they do not even inherit the integer decomposition property, 
as the classical example of polytope without the integer decomposition property $P=\convexhull\left((0,0,0),(1,1,0),(1,0,1),(0,1,1)\right)$ is fully box-integer. To see that $P$ is fully box-integer, note that in the minimal linear description of $P$ given below, the matrix of constraints is totally equimodular.
 Since $P$ is also integer, this implies that $P$ is fully box-integer by Theorem~\ref{TEMPBI}. The point $(1,1,1)$ is in $2P$ and can not be writen as an integer combination of the integer points of $P$, hence $P$ does not have the integer decomposition property.

$$
P=\left\{x\in\mathbb{R}^3:
\begin{bmatrix} 
1 & -1 & -1 \\
-1 & 1 & -1 \\
-1 & -1 & 1 \\
1 & 1 & 1
\end{bmatrix}
x\leq
\begin{bmatrix} 
0 \\
0 \\
0 \\
2
\end{bmatrix}
\right\}
$$
Nevertheless, given the strong integrality properties of fully box-integer polyhedra and as the above large subclass $\mathcal{P}$ has the Integer Carath\'eodory Property, it might be that many of them have the integer decomposition property. 
In this area, a long standing open question is known as Oda's question~\cite{oda}: is it true that every smooth polytope has the integer decomposition property?
%For two neighboring vertices $u,v$ of an integer polytope of $\mathbb{R}^n$, the associated {\em primitive edge vector} is~$\frac{u-v}{\gcd(u-v)}$.
A full-dimensional polytope of $\mathbb{R}^n$ is {\em simple} if every vertex has $n$ neighbors.
A simple integer polytope is {\em smooth} if for every vertex $v$ the generators of the associated minimal tangent cone form a basis of the lattice~$\mathbb{Z}^n$.

The polyhedron of the example above is not smooth, and the following special case of Oda's question is a reasonable first step to determine which fully box-integer polyhedra have the integer decomposition property.

\begin{open}
Do smooth fully box-integer polyhedra have the integer decomposition property?
\end{open}

\subsection{Box-TDIness for Conservative Functions}

In~\cite{cogrla}, the authors prove that the standard system describing the circuit cone is box-TDI if and only if the graph is series-parallel. 
We illustrate that polarity preserves the box-TDIness of cones by providing a box-TDI system for the cone of conservative function---polar of the circuit cone.

\medskip 

Let $G=(V,E)$ be an undirected graph.
The set of edges connecting a given set of vertices and its complement is called a {\em cut}.
A cut containing no other nonempty cut is called a {\em bond}. 
A set of edges is called a {\em circuit} if it induces a connected subgraph where every vertex has degree two.
The {\em minors} of a graph are the graphs obtained by repeatedly contracting edges and deleting edges and isolated vertices.
Given $e\in E$, the graphs obtained from $G$ by respectively deleting and contracting $e$ are denoted by $G\setminus e$ and $G/e$.
A graph is {\em series-parallel} if and only if contains no $K_4$ minor~\cite{du}. 

The {\em circuit cone} $C_{circuit}(G)=\cone\{\chi^C \text{ for all circuits $C$ of $G$}\}$ is the cone generated by the incidence vectors of the circuits of $G$. 
Seymour~\cite{se1} proved that $C_{circuit}(G)=\{x\in\mathbb{R}^E: x\geq \mathbf{0},x(D\setminus e)\geq x_e \text{ for all cuts $D$ of $G$ and $e\in D$}\}$. 
A function $f:E\rightarrow \mathbb R$ is {\em conservative} if $f(C)\geq 0$ for each circuit $C$ of $G$.
These functions form the {\em cone of conservative functions} $C_{cons}(G)=\{x\in\mathbb{R}^E:x(C)\geq 0 \text{ for all circuits $C$ of $G$}\}$. 
By polarity~\cite[Corollary 29.2h]{scbig}, we have $C_{cons}(G)=-C_{circuit}(G)^\star=\cone\{\chi^e \text{ for all $e\in E$}, \chi^{D\setminus e} -\chi^e \text{ for all cuts $D$ of $G$ and $e\in D$}\}$.

\medskip 
We show that box-TDI systems describing $C_{cons}(G)$ only exist when $G$ is series-parallel. In this case, we provide such a system in the following proposition.

\begin{proposition}\label{spsystems}
The system $\frac{1}{2}x(C)\geq 0$ for all circuits $C$ of $G$ is box-TDI if and only if the graph $G$ is series-parallel.
\end{proposition}

\begin{proof}
We first prove that if the graph $G$ is not series-parallel, then its cone of conservative functions is not box-TDI. In this case, no system describing $C_{cons}(G)$ is box-TDI. 
For a graph $G=(V,E)$ and $e\in E$, one can see that $C_{cons}(G\setminus e)$  and $C_{cons}(G/e)$ are respectively obtained by deleting $e$'s coordinate in $C_{cons}(G)\cap \{x_e=+\infty\}$ and $C_{cons}(G)\cap \{x_e=0\}$. Hence, taking minors preserves the box-TDIness of the cone of conservative functions. It remains to prove that $C_{cons}(K_4)$ is not box-TDI. 
Let us apply Theorem~\ref{cor:main}. 

The nonequimodular matrix $M$ of Figure~\ref{figK4} is the constraint matrix obtained by considering the inequalities associated with the three circuits formed by the three internal triangles of $K_4$.
By Observation~\ref{face-def_C}, $M$ is face-defining for $C_{cons}(K_4)$ because $\mathbf{0}$ and the three conservative functions $\chi^4+\chi^5-\chi^1$, $\chi^4+\chi^6-\chi^2$ and $\chi^5+\chi^6-\chi^3$ are affinely independent, belong to $C_{cons}(K_4)$ and satisfy $Mx=\mathbf{0}$.
Therefore, by point~\ref{polyii-} of Theorem~\ref{cor:main}, the cone $C_{cons}(K_4)$ is not box-TDI.

\begin{figure}[h!]
\begin{center}
\hspace{2truecm}\begin{minipage}[c]{.3\textwidth}
\resizebox{.75\textwidth}{!}
{
\begin{tikzpicture}[node distance=1cm,bend angle=45,auto,scale=0.15]
\tikzstyle{vertex}=[draw,thick,circle,scale=0.75,minimum size=1pt]%,fill=black

\node (1) at (0,0) [vertex]{};
\node (2) at (0,10) [vertex]{};
\node (3) at (8.66,-5) [vertex]{};
\node (4) at (-8.66,-5) [vertex]{};

	\path	(1) edge[thick] node[below left]{$1$} (2)
					  edge[thick] node[above]{$3$} (3)
					  edge[thick] node[above]{$2$} (4);
						
	\path	(2) edge[thick] node[right]{$5$} (3)
					  edge[thick] node[left]{$4$} (4);
						
	\path	(3) edge[thick] node[below]{$6$} (4);
\end{tikzpicture}
}
\end{minipage}
\begin{minipage}[c]{.46\textwidth}
$
M=\left[\begin{array}{cccccc}
	1 & 1 & 0 & 1 & 0 & 0\\
	1 & 0 & 1 & 0 & 1 & 0\\
	0 & 1 & 1 & 0 & 0 & 1
\end{array}\right]
$
\end{minipage}
\end{center}
\caption{The graph $K_4$ and a face-defining matrix $M$ of $C_{cons}(K_4)$.}\label{figK4}
\end{figure}
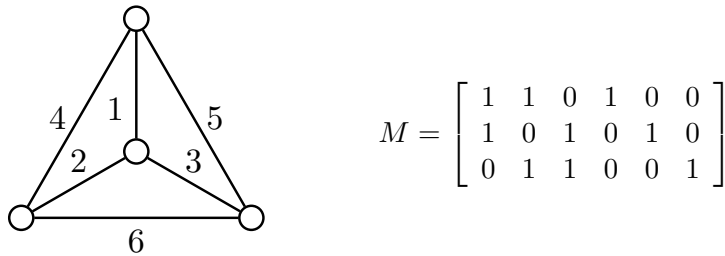

Now, suppose that $G$ is series-parallel. 
Then, \cite[Theorem~1]{cogrla} asserts that the system $x\geq\mathbf{0}$, $x(D\setminus e) \geq x(e)$ for all cuts $D$ of~$G$ and $e\in C$ is box-TDI.
Hence the circuit cone of $G$ is a box-TDI cone. 
By Lemma~\ref{boxtdicone}, $C_{cons}(G)=-C_{circuit}(G)^\star$ is box-TDI. 
By Theorem~\ref{tdiboxtdisystem}, it remains to show that the system $\frac{1}{2}x(C)\geq 0$ for all circuits $C$ of $G$ is TDI. 
\cite[Corollary~22.5a]{sc} states that a system $Ax\leq 0$ is TDI if and only if the rows of $A$ form a Hilbert basis. 
In other words, it remains to show that any integer vector $z$ in the circuit cone of $G$ is a nonnegative integer combination of vectors of~$\mathcal{H}=\{\frac{1}{2}\chi^{C}: \mbox{ $C$ is a circuit of $G$}\}$.
 \cite[Theorem~1]{algozh} asserts that, in graphs with no Petersen minors, if $p$ is an integer vector of the circuit cone such that $p(C)$ is even for all cuts $C$ of $G$, then $p$ is a sum of circuits.
Since the Petersen graph contains a $K_4$ minor, \cite[Theorem~1]{algozh} applies to~$G$. 
Since $2z$ satisfies the conditions, $2z=\sum_{C\in\mathcal{C}}\chi^C$ for some family $\mathcal{C}$ of circuits of~$G$. 
Therefore, $z=\sum_{C\in\mathcal{C}}\frac{1}{2} \chi^C$.
\end{proof}

Note that the coefficients of the system in Proposition~\ref{spsystems} are half-integral. We leave open the question of finding a box-TDI system with integer coefficients, which exists by~\cite[Theorem~22.6(i)]{sc} and Theorem~\ref{tdiboxtdisystem}.

\medskip

By planar duality, there is a correspondance between the circuits of a planar graph and the bonds of its planar dual. 
This is used in~\cite{cogrla} to obtain the box-TDIness of the standard system describing the cut cone of a series-parallel graph. 
Applying planar duality to Proposition~\ref{spsystems} provides the following: if the graph is series-parallel, then $\frac{1}{2}x(B)\geq 0$ for all bonds $B$ is a box-TDI system describing the polar of the cut cone.
This is in fact an equivalence as one can check that the box-TDIness of the corresponding cone is preserved under taking minors and that the matrix of Figure~\ref{figK4} is face-defining when $G=K_4$.

\section*{Acknowledgment}

We are grateful to Andr{\'a}s Seb{\H{o}} for his invaluable comments and suggestions.

\bibliography{../../Bibtex/bibli}

\begin{thebibliography}{10}

\bibitem{algozh}
B.~Alspach, M.~Goddyn, and C.-Q. Zhang.
\newblock Graphs with the circuit cover property.
\newblock {\em Transactions of the American Mathematical Society},
  344:131--154, 1994.

\bibitem{ap}
G.~Appa.
\newblock k-integrality, an extension of total unimodularity.
\newblock {\em Operations Research Letters}, 13:159–163, 1993.

\bibitem{apko}
G.~Appa and B.~Kotnyek.
\newblock Rational and integral k-regular matrices.
\newblock {\em Discrete Mathematics}, 275(1):1--15, 2004.

\bibitem{ba}
S.~Barnett.
\newblock {\em Matrices in control theory}.
\newblock R.E. Krieger, 1984.

\bibitem{be}
Y.~Benchetrit.
\newblock Integer round-up property for the chromatic number of some h-perfect
  graphs.
\newblock {\em Mathematical Programming}, 164(1):245--262, Jul 2017.

\bibitem{beha}
J.~H. Bevis and F.~J. Hall.
\newblock Some classes of integral matrices.
\newblock {\em Linear Algebra and its Applications}, 48:473--483, 1982.

\bibitem{caed}
K.~Cameron.
\newblock Polyhedral and algorithmic ramifications of antichains.
\newblock {\em Ph.D. Thesis, University of Waterloo. (Supervisor: J. Edmonds)},
  1982.

\bibitem{caca}
K.~Cameron.
\newblock A min-max relation for the partial q- colourings of a graph. part ii:
  Box perfection.
\newblock {\em Discrete Mathematics}, 74(1):15--27, 1989.
\newblock Special Double Issue.

\bibitem{ca}
P.~Camion.
\newblock Characterization of totally unimodular matrices.
\newblock {\em Proceedings of the American Mathematical Society},
  16(5):1068--1073, 1965.

\bibitem{chchza}
X.~Chen, Z.~Chen, and W.~Zang.
\newblock A unified approach to box-mengerian hypergraphs.
\newblock {\em Mathematics of Operations Research}, 35(3):655--668, 2010.

\bibitem{chdiza}
X.~Chen, G.~Ding, and W.~Zang.
\newblock A characterization of box-mengerian matroid ports.
\newblock {\em Mathematics of Operations Research}, 33(2):497--512, 2008.

\bibitem{chdiza2}
X.~Chen, G.~Ding, and W.~Zang.
\newblock The box-tdi system associated with 2-edge connected spanning
  subgraphs.
\newblock {\em Discrete Applied Mathematics}, 157(1):118--125, 2009.

\bibitem{co}
W.~Cook.
\newblock On box totally dual integral polyhedra.
\newblock {\em Mathematical Programming}, 34(1):48--61, 1986.

\bibitem{cogrla}
D.~Cornaz, R.~Grappe, and M.~Lacroix.
\newblock Trader multiflow and box-tdi polyhedra.
\newblock {\em Submitted}, 2016.

\bibitem{difeza}
G.~Ding, L.~Feng, and W.~Zang.
\newblock The complexity of recognizing linear systems with certain integrality
  properties.
\newblock {\em Mathematical Programming}, 114(2):321--334, 2008.

\bibitem{ditaza}
G.~Ding, L.~Tan, and W.~Zang.
\newblock When is the matching polytope box-totally dual integral?
\newblock {\em Mathematics of Operations Research}, 43(1):64--99, 2018.

\bibitem{diza}
G.~Ding and W.~Zang.
\newblock Packing cycles in graphs.
\newblock {\em Journal of Combinatorial Theory, Series B}, 86(2):381--407,
  2002.

\bibitem{dizazh}
G.~Ding, W.~Zang, and Q.~Zhao.
\newblock On box-perfect graphs.
\newblock {\em Journal of Combinatorial Theory, Series B}, 128:17--46, 2018.

\bibitem{du}
R.~Duffin.
\newblock Topology of series-parallel networks.
\newblock {\em Journal of Mathematical Analysis and Applications},
  10(2):303--318, 1965.

\bibitem{ed}
J.~Edmonds.
\newblock {\em Submodular Functions, Matroids, and Certain Polyhedra}, pages
  11--26.
\newblock Springer Berlin Heidelberg, Berlin, Heidelberg, 2003.

\bibitem{edgi}
J.~Edmonds and R.~Giles.
\newblock Total dual integrality of linear inequality systems.
\newblock {\em Progress in combinatorial optimization (Jubilee Conference,
  University of Waterloo,Waterloo, Ontario, 1982, W. R. Pulleyblank, ed.)},
  pages 117--129, 1984.

\bibitem{fofu}
L.~R. Ford and D.~R. Fulkerson.
\newblock Maximal flow through a network.
\newblock {\em Canadian Journal of Mathematics}, 8:399–404, 1956.

\bibitem{gela}
A.~Gerards and M.~Laurent.
\newblock A characterization of box 1d-integral binary clutters.
\newblock {\em Journal of Combinatorial Theory, Series B}, 65(2):186--207,
  1995.

\bibitem{ghho}
A.~Ghouila-Houri.
\newblock Caracterisation des matrices totalement unimodulaires.
\newblock {\em Comptes Rendus Hebdomadaires des Seances de I'Academic des
  Sciences (Paris)}, 254:1192--1194, 1962.

\bibitem{gi}
D.~Gijswijt.
\newblock Integer decomposition for polyhedra defined by nearly totally
  unimodular matrices.
\newblock {\em SIAM Journal on Discrete Mathematics}, 19(3):798--806, 2005.

\bibitem{gire}
D.~Gijswijt and G.~Regts.
\newblock Polyhedra with the integer carathéodory property.
\newblock {\em Journal of Combinatorial Theory, Series B}, 102(1):62--70, 2012.

\bibitem{he}
I.~Heller.
\newblock On linear systems with integral valued solutions.
\newblock {\em Pacific Journal of Mathematics}, 7(3):1351--1364, 1957.

\bibitem{heto}
I.~Heller and C.~B. Tompkins.
\newblock An extension of a theorem of dantzig’s.
\newblock In {\em Linear Inequalities and Related Systems. (AM-38)}, pages
  247--254. Princeton University Press, 1956.

\bibitem{hoop}
A.~Hoffman and R.~Oppenheim.
\newblock Local unimodularity in the matching polytope.
\newblock In B.~Alspach, P.~Hell, and D.~Miller, editors, {\em Algorithmic
  Aspects of Combinatorics}, volume~2 of {\em Annals of Discrete Mathematics},
  pages 201--209. Elsevier, 1978.

\bibitem{hokr}
A.~J. Hoffman and J.~B. Kruskal.
\newblock Integral boundary points of convex polyhedra.
\newblock In {\em Linear Inequalities and Related Systems. (AM-38)}, pages
  223--246. Princeton University Press, 1956.

\bibitem{ko}
B.~Kotnyek.
\newblock A generalization of totally unimodular and network matrices.
\newblock {\em Ph.D. Thesis, London School of Economics}, 2002.

\bibitem{le}
J.~Lee.
\newblock Subspaces with well-scaled frames.
\newblock {\em Linear Algebra and its Applications}, 114:21--56, 1989.
\newblock Special Issue Dedicated to Alan J. Hoffman.

\bibitem{lo}
L.~Lov{\'{a}}sz.
\newblock Normal hypergraphs and the perfect graph conjecture.
\newblock {\em Discrete Mathematics}, 306(10-11):867--875, 2006.

\bibitem{oda}
T.~Oda.
\newblock Problems on minkowski sums of convex lattice polytopes.
\newblock {\em arXiv:0812.1418}, 2008.

\bibitem{pa}
J.~Pap.
\newblock Recognizing conic {TDI} systems is hard.
\newblock {\em Mathematical Programming}, 128(1-2):43--48, 2011.

\bibitem{paya}
C.~H. Papadimitriou and M.~Yannakakis.
\newblock On recognizing integer polyhedra.
\newblock {\em Combinatorica}, 10(1):107--109, Mar 1990.

\bibitem{sc}
A.~Schrijver.
\newblock {\em Theory of linear and integer programming}.
\newblock Wiley-Interscience series in discrete mathematics and optimization.
  Wiley, 1999.

\bibitem{scbig}
A.~Schrijver.
\newblock {\em Combinatorial optimization: polyhedra and efficiency.}
\newblock Algorithms and combinatorics. Springer-Verlag, Berlin, Heidelberg,
  New York, N.Y., et al., 2003.

\bibitem{se}
A.~Seb{\H{o}}.
\newblock Path partitions, cycle covers and integer decomposition.
\newblock In M.~Lipshteyn, V.~E. Levit, and R.~M. McConnell, editors, {\em
  Graph Theory, Computational Intelligence and Thought: Essays Dedicated to
  Martin Charles Golumbic on the Occasion of His 60th Birthday}, pages
  183--199. Springer Berlin Heidelberg, Berlin, Heidelberg, 2009.

\bibitem{se2}
P.~Seymour.
\newblock The matroids with the max-flow min-cut property.
\newblock {\em Journal of Combinatorial Theory, Series B}, 23(2):189--222,
  1977.

\bibitem{se1}
P.~Seymour.
\newblock Sums of circuits.
\newblock {\em Graph theory and related topics (Proceedings Conference,
  Waterloo, Ontario, 1977; J.A. Bondy, U.S.R. Murty,}, 1:341--355, 1979.

\bibitem{sm}
H.~J.~S. Smith.
\newblock On systems of linear indeterminate equations and congruences.
\newblock {\em Philosophical Transactions of the Royal Society of London},
  151:293--326, 1861.

\bibitem{tr}
K.~Truemper and R.~Chandrasekaran.
\newblock Local unimodularity of matrix-vector pairs.
\newblock {\em Linear Algebra and its Applications}, 22:65--78, 1978.

\bibitem{veda}
A.~F.~J. Veinott and G.~B. Dantzig.
\newblock Integral extreme points.
\newblock {\em SIAM Review}, 10(3):371--372, 1968.

\bibitem{za}
G.~Zambelli.
\newblock Colorings of k-balanced matrices and integer decomposition property
  of related polyhedra.
\newblock {\em Operations Research Letters}, 35(3):353--356, 2007.

\end{thebibliography}
%\bibliography{bibli}
\bibliographystyle{abbrv}
\end{document}